\documentclass[a4paper,USenglish,cleveref,autoref]{lipics-v2021}

\usepackage{algorithm}
\usepackage{algpseudocode}
\usepackage{float}
\floatname{algorithm}{Rule Set}

\newcommand{\ie}{{\it i.e.,~}}
\newcommand{\eg}{{\it e.g.,~}}

\newcommand{\MVS}{\textsf{MVS}}
\newcommand{\MVN}{\textsf{MVN}}

\newcommand{\calA}{\mathcal{A}}
\newcommand{\calAC}{\mathcal{A}^\mathcal{C}}

\newcommand{\calC}{\mathcal{C}}
\newcommand{\calB}{\mathcal{B}}

\newcommand{\BigO}{\mathcal{O}}

\newcommand{\capa}{\mathsf{cap}}
\newcommand{\Sstar}{S^{*}}

\newcommand{\ccaps}{\mathsf{cap}^{*}}

\title{The Maximum Mutual Visibility Set on a Cactus Graph and the Self-stabilizing Constructions} 

\titlerunning{A Maximum Mutual Visibility Set on a Cactus Graph} 

\author{Yonghwan Kim}
{Nagoya Institute of Technology, Aichi, Japan}
{kim@nitech.ac.jp}{https://orcid.org/0000-0002-5437-7626}{}

\author{Yuichi Sudo}
{Hosei University, Tokyo, Japan}
{sudo@hosei.ac.jp}{https://orcid.org/0000-0002-4442-1750}{}

\authorrunning{Y. Kim and Y. Sudo} 

\Copyright{Yonghwan Kim} 



\category{} 

\relatedversion{} 

\supplement{} 

\acknowledgements{This work was supported by JSPS KAKENHI Grant Numbers JP23K24825, JP25K03078, JP25K03079, JP26K02865, JP26K23809, and JST FOREST Program JPMJFR226U.}

\nolinenumbers 

\begin{document}

\maketitle

\begin{abstract}
Given a graph $G=(V,E)$, let $S$ ($\subseteq V$) be a set of vertices.
Two vertices are \emph{mutually visible} 
if there exists a shortest path in $G$ between them that does not contain any other vertex of $S$.
A set $S$ is a \emph{Mutual Visibility Set} (\MVS) 
if every pair of vertices in $S$ is mutually visible.

The concept of \MVS s in graphs has attracted significant attention
since its introduction, as it provides an important structural property of graphs.
However, determining a maximum \MVS\ in general graphs is computationally intractable; 
the decision problem of whether a graph admits an \MVS\ of size 
at least $k$ has been shown to be \emph{NP-complete}.
Thus, prior work has focused on finding maximal \MVS s 
or restricting attention to specific graph classes.

Cactus graphs form a fundamental low-treewidth class, 
yet the maximum \MVS\ problem for this class remains open.
In this paper, we first determine the size of maximum \MVS~in cactus graphs, 
and introduce two self-stabilizing algorithms that construct such sets.
The first algorithm uses a single BFS tree and stabilizes in $O(D)$
rounds with $O(\log n)$ bits per process on average; the second one
uses parallel BFS trees and stabilizes in
$O(|C_{\max}|+|T_{\max}|)$ rounds, which we show to be asymptotically
tight as a function of these two parameters, even on graphs where
$|C_{\max}|+|T_{\max}| = o(D)$.
\end{abstract}

\section{Introduction}
\label{sec:intro}

Distributed systems must cope with the absence of a global clock, only
partial knowledge of the global state, and transient faults that
silently drive the system into an illegitimate configuration.
A \emph{self-stabilizing} algorithm~\cite{Dijkstra1974} tolerates such
faults by guaranteeing convergence to a legitimate configuration
within finite time from \emph{any} initial configuration, without
external intervention.

A natural structural question in distributed network settings is
which subsets of processors can communicate efficiently without
mutual interference.
This is captured by the notion of a \emph{Mutual Visibility Set}
(\MVS), introduced by Di Stefano~\cite{DiStefano2022}.
Given a graph $G = (V, E)$, a set $S \subseteq V$ is an \MVS\ if
every pair of vertices in $S$ is connected by a shortest path whose
internal vertices all lie outside $S$.
In a distributed system modeled by $G$, a maximum \MVS\ is the largest
set of processes that can be connected pairwise by shortest paths none
of which is routed \emph{through} another member of the set.
This is the natural formalisation of a network of mutually
distrusting parties---each party is willing to relay traffic for
others, but is not willing to have its own traffic relayed by a peer,
so that any two members must be joined by a shortest path whose
intermediate vertices are all non-members.
The same condition arises when the members of the set are the
end-points of concurrent shortest-path communications and a member
must not be forced to act as a relay while it is itself communicating.
Mutual visibility originates in the study of mobile robots with
obstructed visibility, where a robot blocks the line of sight of two
robots collinear with it~\cite{DiLunaEtAl2017}, and the graph version
studied here replaces lines of sight by shortest
paths~\cite{DiStefano2022}.

The \emph{mutual visibility number} (\MVN), 
denoted $\mu(G)$, is the size of the largest \MVS\ on the given graph.
Determining a maximum \MVS\ or finding \MVN\ 
is computationally intractable on general
graphs~\cite{DiStefano2022} and hard to approximate for large
diameter~\cite{BiloEtAl2024}.
Motivated by these hardness barriers, prior work has restricted
attention to specific graph classes, obtaining exact formulas for
paths, cycles, block graphs, cographs, grids, and
distance-hereditary graphs~\cite{DiStefano2022,CiceroneDistefano2023dh},
as well as various graph products~\cite{CiceroneEtAl2023,KorzeVesel2025}.
Despite this sustained activity, the \MVN\ of \emph{cactus graphs}---a
fundamental class of low-treewidth graphs that generalizes both trees
and cycles---has not been determined.

In this paper we resolve the case of cactus graphs.
We establish the exact \MVN\ of any cactus graph via a closed-form
formula, and propose two self-stabilizing algorithms that construct a
maximum \MVS.
Our contributions are as follows.
\begin{enumerate}
   \item We introduce a cycle--tree decomposition of a cactus graph
          $G$ and classify each cycle component as a zero-cycle,
          unit-cycle, twin-cycle, or latent twin-cycle, leading to
          the exact formula
          $\mu(G) \;=\; 2n_2 + n_{\mathrm{lt}} + n_1 + \ell(G)$,
          where $n_1$, $n_2$, $n_{\mathrm{lt}}$, and $\ell(G)$ denote
          the numbers of unit-cycles, twin-cycles, latent twin-cycles,
          and leaves of $G$, respectively.
          No exact characterization of the \MVN\ for cactus graphs was previously known.
    \item We propose \textbf{Algorithm~1}, a self-stabilizing algorithm
          that constructs a maximum \MVS\ using a single-root BFS tree.
          It stabilizes in $O(D)$ rounds and uses $O(\log n)$ bits per
          process on \emph{average}, where $D$ is the diameter of $G$
          and $n=|V|$.
    \item We propose \textbf{Algorithm~2}, a self-stabilizing algorithm
          that runs parallel BFS trees from all high-degree vertices,
          enabling component-wise early termination.
          It stabilizes in $O(|C_{\max}|+|T_{\max}|)$ rounds, where
          $|C_{\max}|$ and $|T_{\max}|$ are the sizes of the largest
          cycle component and of the largest tree component of $G$.
          This bound is \emph{component-local}: we exhibit a family of
          cactus graphs on which $|C_{\max}|+|T_{\max}|$ is an
          arbitrarily small fraction of $D$ and on which every correct
          self-stabilizing algorithm nevertheless requires
          $\Omega(|C_{\max}|+|T_{\max}|)$ rounds, so the bound is
          asymptotically tight as a function of these two parameters.
\end{enumerate}
Both algorithms operate under the distributed daemon---the weakest
standard scheduling assumption---and are guaranteed to converge to
a maximum \MVS\ from any initial configuration.

The rest of the paper is organized as follows.
Section~\ref{sec:related} reviews related work.
Section~\ref{sec:mvs} derives the \MVN\ formula for cactus graphs and
gives all of its proofs.
Section~\ref{sec:algorithm} fixes the computational model, describes
the two self-stabilizing algorithms, gives the complete set of rules
of Algorithm~1 together with their correctness and complexity
analysis, and proves the lower bound that makes the stabilization time
of Algorithm~2 asymptotically tight.
Section~\ref{sec:conclusion} concludes with open problems.
This is the full version of the paper: every proof is given in place,
and the parts that a page-limited presentation can only summarise ---
the rules of Algorithm~1, the layered analysis on which the round
complexity rests, and the construction of the lower-bound family ---
are developed here in detail.

\section{Related Work}
\label{sec:related}
Self-stabilization was introduced by Dijkstra in
1974~\cite{Dijkstra1974} as a paradigm for designing fault-tolerant
distributed systems, and self-stabilizing solutions are now known for
most classical distributed problems; see Dolev~\cite{Dolev2000} and
Altisen et al.~\cite{Altisen2019} for textbook treatments,
and~\cite{Hedetniemi2010} for a survey of graph-theoretic problems.
Two results from this line are directly relevant here.
First, for \emph{silent} algorithms---those whose communication
registers stop changing after stabilization---Dolev, Gouda and
Schneider~\cite{DolGouSch1999} proved an $\Omega(\log n)$ lower bound
on the memory per register for problems such as leader election and
spanning tree construction; Remark~\ref{rem:space-lb} discusses why it
does not transfer to the maximum \MVS\ problem.
Second, stronger fault models such as
snap-stabilization~\cite{SnapStab},
superstabilization~\cite{DolHer1997} and Byzantine-tolerant
stabilization~\cite{DuboisMasuzawaTixeuil2010} have been developed;
we work in the classical model under the distributed daemon.

Visibility problems in distributed computing originate in mobile robot
systems: Di Luna et al.~\cite{DiLunaEtAl2017} solved the mutual
visibility problem for anonymous robots with obstructed visibility in
the plane.
The graph-theoretic formulation of the \MVS\ problem was introduced by
Di Stefano~\cite{DiStefano2022}, who proved that deciding whether a
graph admits an \MVS\ of size at least $k$ is \emph{NP-complete} and
determined the \MVN\ of several graph classes, including block graphs,
trees, Cartesian products of paths and cycles, complete bipartite
graphs, and cographs.
Bil\`{o} et al.~\cite{BiloEtAl2024} proved strong inapproximability
results: the problem is not approximable within $n^{1/3-\varepsilon}$
for graphs of diameter at least three and is APX-hard for diameter
two.

Because the general problem is intractable, subsequent work has either
restricted the graph class or introduced variants of mutual
visibility.
Exact formulas for the \MVN\ are known for paths, cycles, block
graphs, cographs and grids, and the \MVN\ is computable in linear time
on distance-hereditary graphs~\cite{CiceroneDistefano2023dh}; further
results cover Cartesian products and triangle-free
graphs~\cite{CiceroneEtAl2023}, strong
products~\cite{CiceroneEtAl2023strong}, graphs of diameter
two~\cite{CiceroneEtAl2024}, hypercubes~\cite{KorzeVesel2025}, and
Sierpi\'{n}ski-type graphs~\cite{SierpMV}.
The \emph{total}, \emph{dual} and \emph{outer} variants were proposed
in~\cite{CiceroneEtAl2023variety}, and an independent variant
in~\cite{BrevarYero2024}.
The \MVN\ of \emph{cactus graphs} 
has not been determined; the present paper settles it
and, in addition, gives two self-stabilizing
algorithms that construct a maximum \MVS.

\section{The Mutual Visibility Number on Cactus Graphs}
\label{sec:mvs}

This section is organised as follows.
Section~\ref{sec:prelim} fixes the terminology and recalls the two
known results on which our analysis rests.
Section~\ref{sec:geometry} collects the elementary metric facts about a
single cycle component, and in particular introduces the two parameters
$\lambda_C$ and $\tau_C$.
Section~\ref{sec:types} uses them to classify cycle components and to
define the \emph{capacity} $\capa(\cdot)$, which is the quantity that
the main theorem computes.
Section~\ref{sec:lower} exhibits an explicit \MVS\ $\Sstar$ of size
$\capa(G)$ (the lower bound), and Section~\ref{sec:upper} proves the
matching upper bound and the exact formula.

All proofs are given in full and in place; nothing is deferred.

\subsection{Preliminaries}
\label{sec:prelim}

All graphs considered in this paper are finite, simple, and undirected.
Let $G = (V, E)$ be a connected graph
with vertex set $V = \{v_1, v_2, \ldots, v_n\}$
and edge set $E$.
For a vertex $v_i \in V$, we denote by $N(v_i)$
the \emph{open neighborhood} of $v_i$,
\ie $N(v_i) = \{v_j \in V \mid \{v_j, v_i\} \in E\}$.
We denote by $\delta(v_i)$ the degree of vertex $v_i$,
\ie the number of its adjacent vertices.
A vertex of degree one is a \emph{leaf} of $G$, and we write
$\ell(G)$ for the number of leaves of $G$; more generally, for a
subgraph $H$ of $G$ we write $\ell_G(H)$ for the number of leaves of
$G$ that lie in $H$, so that $\ell(G) = \ell_G(G)$.
We write $d_G(u,v)$ for the distance between $u$ and $v$ in $G$,
\ie the length of a shortest $u$-$v$ path in $G$.

The indices of $v_1,\dots,v_n$ are used only to break ties in an
otherwise arbitrary but fixed way; no property of the enumeration is
assumed.

\begin{definition}[Cactus Graph]
\label{def:cactus}
A connected graph $G = (V, E)$ is a \emph{cactus graph}
if every edge of $G$ belongs to at most one simple cycle.
Equivalently, $G$ is a cactus graph if and only if
every biconnected component of $G$
is either a single edge or a simple cycle.
\end{definition}

\begin{figure}[tb]
  \begin{center}
    \includegraphics[width=\linewidth]{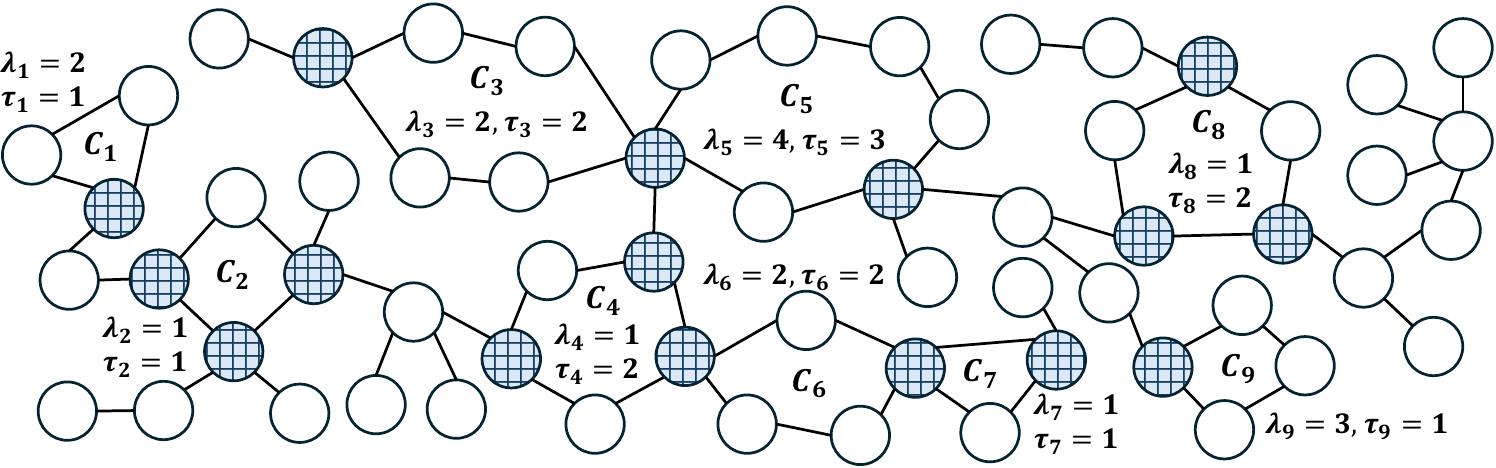}
    \caption{An example of a cactus graph and its cycle components with parameters.}
    \label{Fig:cactuscycle}
  \end{center}
\end{figure}

Figure~\ref{Fig:cactuscycle} illustrates an example of a cactus graph.
Note that simple cycles are special cases of cactus graphs,
and the \MVN\ of cycles can be easily obtained~\cite{DiStefano2022}.
Therefore, we focus on cactus graphs that are not simple cycles.
We denote by $\calC(G)$ the set of all simple cycles in $G$,
and write $|\calC(G)|$ for the number of cycles in $G$.

Moreover, we write $v_i \in e_j$ (where $v_i \in V$ and $e_j \in E$)
if $v_i$ is an endpoint of $e_j$.

\begin{definition}[Articulation Point]
\label{def:ap}
Let $G = (V, E)$ be a connected graph.
A vertex $v \in V$ is an \emph{articulation point} of $G$
if the subgraph
$G - v = (V \setminus \{v\},\, E \setminus \{e \in E \mid v \in e\})$
is disconnected.
We denote by $\calA(G)$ the set of all articulation points of $G$.
\end{definition}

We now define a canonical decomposition of a cactus graph into
\emph{cycle components} and \emph{tree components}
by cutting at articulation points that lie on cycles.

\begin{definition}[Cycle--Tree Decomposition]
\label{def:decomposition}
Let $G$ be a cactus graph, and let
$\calAC(G) = \{v \in \calA(G) \mid v \text{ belongs to at least one
cycle in } \calC(G)\}$
be the set of articulation points of $G$ that lie on at least one cycle.
The \emph{cycle--tree decomposition} of $G$ consists of the following
two types of components.
\begin{enumerate}
    \item For each simple cycle $C_i \in \calC(G)$, the
          \emph{cycle component} induced by $C_i$ is the subgraph
          $G[V(C_i)]$.
          Every vertex in $V(C_i) \cap \calAC(G)$ is called a
          \emph{boundary vertex} of cycle component $C_i$.

    \item Let $G' = G - E_{\mathrm{cyc}}$, where
          $E_{\mathrm{cyc}} = \bigcup_{C \in \calC(G)} E(C)$
          denotes the set of all edges belonging to at least one cycle.
          Each connected component of $G'$ with at least two vertices
          is called a \emph{tree component} of $G$.
          Every vertex of a tree component $T_i$ that belongs to
          $\calAC(G)$ is called a \emph{boundary vertex} of $T_i$.
\end{enumerate}
The boundary vertices serve as \emph{attachment points} shared between
a cycle component and one or more adjacent components,
and are included in both the relevant cycle component and the
adjacent component(s).
\end{definition}

In Figure~\ref{Fig:cactuscycle},
the vertices in $\calAC(G)$---that is, articulation points that belong
to at least one cycle---are depicted as grid-patterned vertices.
Note that these vertices are the boundary vertices of some components.
The cactus graph in Figure~\ref{Fig:cactuscycle} has 9 cycle components.
Since every edge of a tree component is a bridge of $G$, every vertex
that is internal to a tree component (\ie has degree at least two in
that component) is an articulation point of $G$; we use this fact
repeatedly.

\begin{definition}[Beard and Non-Beard Tree Component]
\label{def:beard}
Let $G$ be a cactus graph, and let $T$ be a tree component of $G$
in its cycle--tree decomposition (Definition~\ref{def:decomposition}).
We say that $T$ is a \emph{beard} of $G$ if the following three
conditions are all satisfied:
\begin{enumerate}
    \item[(B1)] $T$ is a path graph, \ie $T \cong P_k$ for some
                $k \geq 2$;
    \item[(B2)] exactly one endpoint of $T$ is a boundary vertex of
                $T$, and that boundary vertex is shared with exactly
                one cycle component of $G$; and
    \item[(B3)] the other endpoint of $T$ is a leaf of $G$
                (\ie has degree one in $G$), and every internal
                vertex of $T$ (\ie every vertex of degree two in $T$)
                is not a boundary vertex of any component in the
                cycle--tree decomposition of $G$.
\end{enumerate}
A tree component that does not satisfy all three conditions is called
a \emph{non-beard tree component} of $G$.
\end{definition}

\begin{remark}
\label{rem:beard}
A beard is a pendant path hanging off exactly one cycle, with its free
end a leaf of $G$ and no internal vertex shared with another
component; the smallest beard is $P_2$.
A non-beard tree component either branches (violating (B1)), has two
boundary vertices (violating (B2)), or has an internal vertex that is
an attachment point of another component (violating (B3)).
\end{remark}

Finally, we recall the two results of Di Stefano~\cite{DiStefano2022}
that our analysis builds upon.

\begin{lemma}[\cite{DiStefano2022}]
\label{lem:cycle-mvn}
For every simple cycle $C$ with $|C| \geq 3$ we have $\mu(C) = 3$;
that is, no four vertices of $C$ are pairwise mutually visible, and
some three of them always are.
\end{lemma}

\begin{lemma}[\cite{DiStefano2022}]
\label{lem:no-ap}
Every graph $G$ admits a maximum \MVS\ that contains no articulation
point of $G$.
\end{lemma}

By Lemma~\ref{lem:no-ap} we may, and from now on always do, restrict
attention to \MVS s that avoid $\calA(G)$.

\subsection{Basic Geometry of a Cycle Component}
\label{sec:geometry}

Throughout the remainder of Section~\ref{sec:mvs}, $G$ denotes a cactus
graph that is not a simple cycle.
We first fix the two pieces of notation that the whole analysis is
phrased in --- the \emph{branch} at a boundary vertex and the
parameters $\lambda_C^{X}$, $\tau_C$ --- and then record the metric
facts they obey.

\begin{definition}[Branch]
\label{def:branch}
Let $C$ be a cycle component of $G$ and $v \in V(C)\cap\calAC(G)$.
The \emph{branch of $C$ at $v$}, denoted $G_v$, is the union of all
connected components of $G-v$ that do not contain $C\setminus\{v\}$.
\end{definition}

A union of components is taken so that the definition remains
meaningful when $v$ has several branches, \eg when $v$ lies on two
cycles or carries two pendant paths.
Writing $H_v$ for the component of $G-v$ containing $C\setminus\{v\}$,
we obtain the partition
\begin{equation}
\label{eq:branch-partition}
   V = V(G_v)\ \sqcup\ \{v\}\ \sqcup\ V(H_v).
\end{equation}

The branches of one \emph{fixed} cycle component at its distinct
boundary vertices are pairwise disjoint, being unions of distinct
components of $G$ minus those vertices; and a branch determines the
pair $(C,v)$ it comes from (Lemma~\ref{lem:bookkeeping}(i)).
The family of \emph{all} branches is, however, not laminar: if $v$
lies on two cycles $C,C'$ and also carries a pendant path $T$, then
the branch of $C$ at $v$ and the branch of $C'$ at $v$ both contain
$T$ and neither contains the other.
Laminarity does hold for the sub-family of branches that contain no
vertex of a given \MVS\ (Lemma~\ref{lem:bookkeeping}(i)), which is all
we shall need.

\begin{definition}[Parameters $\lambda_C^{X}$ and $\tau_C$]
\label{def:parameters}
Let $C$ be a cycle component of $G$ and let
$A_C = V(C)\cap\calAC(G)$ be its set of boundary vertices.
For $X \subseteq A_C$ we let $\lambda_C^{X}$ denote the maximum number
of consecutive vertices of $C$ that belong to
$\bigl(V(C)\setminus A_C\bigr) \cup X$;
that is, the length of the longest run of consecutive
non-articulation-point vertices of $C$ once the vertices of $X$ have
been \emph{deactivated} (\ie are treated as non-articulation points).
Such a maximal run is called an \emph{$X$-free run} of $C$.
We abbreviate $\lambda_C = \lambda_C^{\emptyset}$ and
$\lambda_C^{v} = \lambda_C^{\{v\}}$.
Finally we set
$\tau_C = \bigl\lfloor (|C|-1)/2 \bigr\rfloor$,
so that $2\tau_C+1 \leq |C| \leq 2\tau_C+2$.
\end{definition}

Figure~\ref{Fig:cactuscycle} shows the parameters $\lambda_C$
and $\tau_C$ for each cycle component.

\begin{observation}
\label{obs:isometric}
Every cycle component $C$ of $G$ is an isometric subgraph of $G$:
for all $u,v \in V(C)$ we have $d_G(u,v) = d_C(u,v)$, and every
shortest $u$--$v$ path of $G$ is one of the two $u$--$v$ arcs of $C$.
\end{observation}

\begin{proof}
Since $G$ is a cactus, $C$ is a biconnected component (block) of $G$;
two distinct blocks share at most one vertex, so any path leaving
$C$ at a vertex $x$ must re-enter $C$ at the same vertex $x$.
Hence no path between two vertices of $C$ is shorter than the shorter
of the two arcs of $C$.
\end{proof}

\begin{observation}
\label{obs:tree}
Let $S$ be an \MVS\ of $G$ with $S \cap \calA(G) = \emptyset$ and let
$T$ be a tree component of $G$.
Then $S \cap V(T)$ consists of leaves of $G$ only.
Moreover, the internal vertices of any path inside $T$ are
articulation points of $G$, hence they never belong to $S$.
\end{observation}
 
\begin{proof}
Let $u \in V(T)$ not be a leaf of $G$.
If $u$ has degree at least two in $T$, then, since every edge of $T$
is a bridge of $G$, deleting $u$ disconnects $G$ and $u \in \calA(G)$.
Otherwise $u$ has degree one in $T$ but degree at least two in $G$, so
$u$ is incident with an edge of a cycle; hence $u \in V(T) \cap
\calAC(G)$ is a boundary vertex of $T$ and again $u \in \calA(G)$.
In both cases $u \notin S$.
An internal vertex of a path inside $T$ has degree at least two in
$T$ and is covered by the first case.
\end{proof}

\begin{lemma}[Arc geometry]
\label{lem:arc}
Let $C$ be a cycle component of $G$, let $A \subseteq V(C)$ with
$A \neq \emptyset$, and let $P$ be a maximal run of consecutive
vertices of $V(C)\setminus A$, say $|P| = k$.
Let $p,q \in A$ be the two vertices of $A$ adjacent to the two ends
of $P$ (possibly $p = q$, which happens exactly when $|A| = 1$).
Then the $p$--$q$ arc through $P$ has length $k+1$, the complementary
arc has length $|C|-k-1$, and
\begin{enumerate}
  \item[(i)] if $k < \tau_C$, the arc through $P$ is the \emph{unique}
        shortest $p$--$q$ path of $G$;
  \item[(ii)] if $k \geq \tau_C$, the complementary arc is a shortest
        $p$--$q$ path of $G$.
\end{enumerate}
Moreover, for every $u \in P$ at distance $j$ from $p$ along $P$
(so $1 \leq j \leq k$), if $j \le \tau_C$ then the sub-path of $P$
from $p$ to $u$ is the unique shortest $p$--$u$ path of $G$.
\end{lemma}

\begin{proof}
By Observation~\ref{obs:isometric} the only $p$--$q$ paths that can be
shortest are the two arcs, whose lengths are $k+1$ and $|C|-k-1$.
The arc through $P$ is the unique shortest one if and only if
$k+1 < |C|-k-1$, \ie $2k+2 < |C|$, and the complementary arc is a
shortest one if and only if $|C| \le 2k+2$.
Since $2\tau_C+1 \leq |C| \leq 2\tau_C+2$, we have
$2k+2 < |C| \iff k < \tau_C$ and $|C| \le 2k+2 \iff k \ge \tau_C$,
which gives (i) and (ii).
For the last claim, the two $p$--$u$ arcs have lengths $j$ and
$|C|-j$, and $j \le \tau_C$ gives $2j \le 2\tau_C < |C|$.
\end{proof}

Lemma~\ref{lem:arc} explains the threshold $\tau_C$: a run is
\emph{blocking} when it is strictly shorter than $\tau_C$,
\emph{neutral} when it equals $\tau_C$, and \emph{free} when it exceeds
$\tau_C$.

\subsection{Cycle Types and Capacity}
\label{sec:types}

\begin{definition}[Cycle Component Types]
\label{def:cycle-types}
Let $C$ be a cycle component of $G$.
A boundary vertex $v \in A_C$ is a \emph{beard boundary vertex} of $C$
if the branch $G_v$ (Definition~\ref{def:branch}) is a beard of $G$
(Definition~\ref{def:beard}).
We classify $C$ as follows.
\begin{itemize}
  \item $C$ is a \emph{twin-cycle} if $\lambda_C > \tau_C$ (refer to Figure \ref{fig:types}(a));
  \item $C$ is a \emph{unit-cycle} if $\lambda_C = \tau_C$ (refer to Figure \ref{fig:types}(b));
  \item $C$ is a \emph{latent twin-cycle} if $\lambda_C < \tau_C$ and
        $\lambda_C^{v^*} > \tau_C$ for some beard boundary vertex
        $v^*$ of $C$, which we call a \emph{witness} for $C$ (refer to Figure \ref{fig:types}(c));
  \item $C$ is a \emph{zero-cycle} otherwise (refer to Figure \ref{fig:types}(d)).
\end{itemize}
The four types are mutually exclusive and exhaustive.
We denote by $n_0,n_1,n_2,n_{\mathrm{lt}}$ the numbers of zero-,
unit-, twin- and latent twin-cycles of $G$, so that
$n_0+n_1+n_2+n_{\mathrm{lt}} = |\calC(G)|$.
\end{definition}

\begin{figure}[tb]
\centering
\begin{tabular}{@{}cccc@{}}
 \includegraphics[width=.225\linewidth]{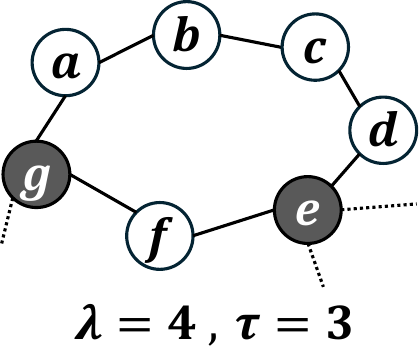} &
 \includegraphics[width=.225\linewidth]{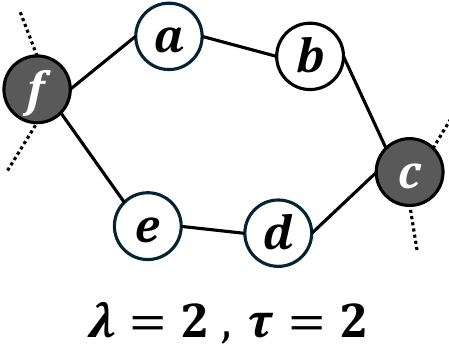} &
 \includegraphics[width=.225\linewidth]{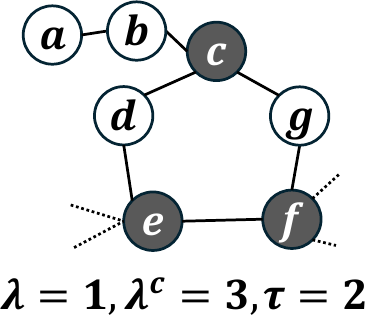} &
 \includegraphics[width=.225\linewidth]{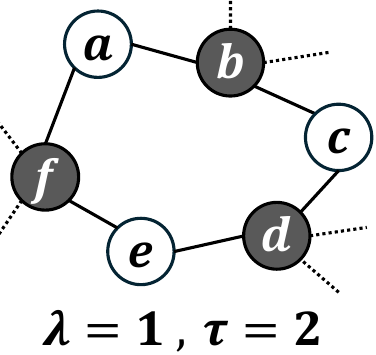} \\
 {\small (a) twin-cycle} & {\small (b) unit-cycle} &
 {\small (c) latent twin-cycle} & {\small (d) zero-cycle}
\end{tabular}
\caption{The four types of cycle component.  Grey vertices are
boundary vertices.}
\label{fig:types}
\end{figure}

\begin{definition}[Capacity]
\label{def:capacity}
The \emph{capacity} of a cycle component $C$ is
\[
  \capa(C) \;=\;
  \begin{cases}
     2 & \text{if $C$ is a twin-cycle},\\
     1 & \text{if $C$ is a unit-cycle or a latent twin-cycle},\\
     0 & \text{if $C$ is a zero-cycle},
  \end{cases}
\]
\end{definition}

\begin{definition}[Branch capacity]
\label{def:branchcap}
Let $C$ be a cycle component of $G$ and $v \in A_C$.
We say that a cycle component $C'$ of $G$ \emph{lies below} $(C,v)$ if
$V(C') \subseteq V(G_v) \cup \{v\}$, and we write $\calC(C,v)$ for the
set of such cycle components.
The \emph{branch capacity} of $C$ at $v$ is
\[
   \ccaps(C,v) \;=\; \ell_G(G_v) \;+\!\!
   \sum_{C' \in \calC(C,v)}\!\! \capa(C').
\]
Finally, the capacity of $G$ is
$\capa(G) = \ell(G) + \sum_{C \in \calC(G)} \capa(C)$.
\end{definition}

Note that $C\notin\calC(C,v)$, while every cycle component through $v$
other than $C$ \emph{does} lie below $(C,v)$; counting $v$ itself is
essential.
For instance, if $G$ consists of two $4$-cycles $C,C'$ sharing a
single vertex $v$, then $G_v=C'\setminus\{v\}$ contains no leaf of $G$
and no cycle component as a subgraph, yet $\ccaps(C,v)=\capa(C')=2$,
which is the correct value.

The capacity of $C$ is what $C$ contributes \emph{net} to a maximum
\MVS: a latent twin-cycle contributes two vertices of $C$ but forfeits
the leaf of one adjacent beard, hence $2-1=1$.
Summing over all components gives
\begin{equation}
\label{eq:capG}
  \capa(G) \;=\; \ell(G) + \sum_{C\in\calC(G)}\capa(C)
           \;=\; 2n_2 + n_{\mathrm{lt}} + n_1 + \ell(G),
\end{equation}
and Theorem~\ref{thm:mvn} will state that $\mu(G) = \capa(G)$.
The following two lemmas make the classification well behaved.

\begin{lemma}
\label{lem:unique-run}
Let $C$ be a cycle component and $X \subseteq A_C$ with
$\lambda_C^{X} > \tau_C$.
Then the $X$-free run of length $\lambda_C^{X}$ is unique.
Moreover, if $C$ is a latent twin-cycle with witness $v^*$ and
$X = \{v^*\}$, then $v^*$ is an \emph{internal} vertex of that run;
in particular both endpoints of the run are non-articulation points
of $G$.
\end{lemma}

\begin{proof}
Two distinct $X$-free runs of length $\lambda > \tau_C$ are separated
by at least one vertex on each side, hence
$|C| \geq 2\lambda + 2 \geq 2\tau_C + 4 > 2\tau_C+2 \geq |C|$,
a contradiction.

For the second claim, let $\alpha$ and $\beta$ be the lengths of the
two $\emptyset$-free runs of $C$ adjacent to $v^*$
(with $\alpha=0$, resp.\ $\beta=0$, if $v^*$ is adjacent to another
boundary vertex on that side).
Deactivating $v^*$ merges exactly these two runs, so
$\lambda_C^{v^*} = \max\{\lambda_C,\ \alpha+1+\beta\}$; since
$\lambda_C < \tau_C < \lambda_C^{v^*}$ we must have
$\alpha+1+\beta = \lambda_C^{v^*} > \tau_C$.
The vertex $v^*$ is an endpoint of the merged run if and only if
$\alpha = 0$ or $\beta = 0$.
If $\alpha = 0$, then $\beta > \tau_C - 1$, \ie $\beta \geq \tau_C$,
whereas $\beta \le \lambda_C < \tau_C$ --- a contradiction; the case
$\beta = 0$ is symmetric.
Hence $\alpha,\beta \geq 1$ and $v^*$ is internal.
Every vertex of the run other than $v^*$ lies outside $A_C$ and is
therefore not an articulation point of $G$, so in particular the two
endpoints are not.
\end{proof}

\begin{lemma}[Value of a branch]
\label{lem:branch-value}
Let $C$ be a cycle component of $G$ and $v \in A_C$.
Then $\ccaps(C,v) \geq 1$, and $\ccaps(C,v) = 1$ holds if and only if
$G_v$ is a beard of $G$.
\end{lemma}

\begin{proof}
We argue by induction on $|V(G_v)|$, the induction hypothesis being
that the statement holds for every branch with strictly fewer
vertices.
Note first that $G_v \neq \emptyset$, since $v$ is an articulation
point and $H_v$ is only one of the components of $G-v$.
 
\medskip\noindent
\textbf{Case 1: $\calC(C,v) = \emptyset$.}
Then $v$ lies on no cycle other than $C$, and $G_v$ contains no cycle
of $G$; hence $G_v$ contains no edge of $E_{\mathrm{cyc}}$ and
$T := G[V(G_v)\cup\{v\}]$ is exactly the tree component of $G$
containing $v$.
Every component of $G_v$ is a tree, and a vertex of it at maximum
distance from $v$ has degree one in $G$; hence
$\ccaps(C,v) = \ell_G(G_v) \geq 1$.
Equality holds if and only if $T$ contains exactly one leaf of $G$,
i.e.\ if and only if $T$ is a path with $v$ as one endpoint and a leaf
of $G$ as the other.
In that case (B1) holds, (B2) holds because $v$ lies on the single
cycle $C$, and (B3) holds because an internal vertex of $T$ lying in
$\calAC(G)$ would lie on a cycle contained in $G_v$, contradicting
$\calC(C,v) = \emptyset$.
Hence $\ccaps(C,v)=1$ if and only if $G_v$ is a beard.
 
\medskip\noindent
\textbf{Case 2: $\calC(C,v) \neq \emptyset$.}
Choose $C' \in \calC(C,v)$ at minimum distance from $v$ and let $x$ be
the vertex of $C'$ closest to $v$; thus $x = v$ when $v \in V(C')$.
In either case $x$ is an articulation point of $G$ lying on $C'$
(if $x \neq v$ it separates $C'$ from $v$; if $x = v$ it separates
$C'$ from $C$), so $x \in A_{C'}$.
Put $k' = |A_{C'}|$.
\begin{itemize}
  \item \textbf{$k'=1$:} then $A_{C'} = \{x\}$, so $\lambda_{C'} = |C'|-1 >
        \tau_{C'}$, i.e.\ $C'$ is a twin-cycle and
        $\ccaps(C,v) \geq \capa(C') = 2$.
  \item \textbf{$k'=2$:} write $A_{C'} = \{x,w\}$.
        The two $\emptyset$-free runs of $C'$ have lengths $a-1$ and
        $|C'|-a-1$ for some $a$, so
        $\lambda_{C'} \geq \lceil |C'|/2\rceil - 1 = \tau_{C'}$ and
        $\capa(C') \geq 1$.
        Let $G_w$ be the branch of $C'$ at $w$.
        Since $w \in V(C') \subseteq V(G_v)\cup\{v\}$ and $w \neq x$,
        we have $w \in V(G_v)$; moreover the component of $G-w$
        containing $C'\setminus\{w\}$ contains $x$ and hence $v$, so
        $v \notin V(G_w)$ and therefore
        $V(G_w)\cup\{w\} \subseteq V(G_v)$.
        Consequently $\ell_G(G_w) \le \ell_G(G_v)$ and
        $\calC(C',w) \subseteq \calC(C,v)$, while
        $C' \notin \calC(C',w)$.
        As $|V(G_w)| < |V(G_v)|$, the induction hypothesis gives
        $\ccaps(C',w) \geq 1$ and hence
        $\ccaps(C,v) \geq \capa(C') + \ccaps(C',w) \geq 2$.
  \item \textbf{$k'\geq 3$:} $C'$ has two boundary vertices $w_1 \neq w_2$,
        both different from $x$.
        As above $V(G_{w_i})\cup\{w_i\} \subseteq V(G_v)$, and
        $G_{w_1}, G_{w_2}$ are disjoint because they are branches of
        the same cycle component at distinct boundary vertices;
        since $w_1 \notin V(G_{w_2})$ and $w_2 \notin V(G_{w_1})$, the
        sets $\calC(C',w_1)$ and $\calC(C',w_2)$ are disjoint as well.
        By the induction hypothesis
        $\ccaps(C,v) \geq \ccaps(C',w_1) + \ccaps(C',w_2) \geq 2$.
\end{itemize}
In all three cases $\ccaps(C,v) \geq 2$, and $G_v$ is not a beard
because a beard contains no cycle and $v$ would then lie on the single
cycle $C$.
The two cases together prove both assertions.
\end{proof}

\begin{remark}
\label{rem:whybeard}
Lemma~\ref{lem:branch-value} says that deactivating a boundary vertex
costs at least one, and at least two unless its branch is a beard.
Since the gain on a cycle component never exceeds $2$ while one of its
boundary vertices stays active (Lemma~\ref{lem:cycle-bound}(iii)),
deactivating a non-beard branch or two or more branches can never pay
off; this is why only beard boundary vertices occur in
Definition~\ref{def:cycle-types}.
The remaining case, in which all boundary vertices are deactivated, is
handled quantitatively in Lemma~\ref{lem:exchange}.
\end{remark}

\subsection{The Canonical Set: the Lower Bound}
\label{sec:lower}

\begin{definition}[Canonical set $\Sstar$]
\label{def:canonical}
For every latent twin-cycle $C$ fix \emph{one} witness $v^*_C$
(Definition~\ref{def:cycle-types}), say the one of smallest index, and
let $z_C$ be the leaf of the beard $G_{v^*_C}$.
For every unit-cycle $C$ fix the $\emptyset$-free run $P_C$ with
$|P_C| = \lambda_C = \tau_C$ containing the vertex of smallest index
among all such runs, and let $u_C$ be the vertex of smallest index
in $P_C$.
For every twin-cycle (resp.\ latent twin-cycle) $C$ let $P_C$ be the
unique longest $\emptyset$-free (resp.\ $\{v^*_C\}$-free) run, which is
unique by Lemma~\ref{lem:unique-run}, and let $a_C,d_C$ be its two
endpoints.
We define
\begin{align*}
  \Sstar \;=\;&
  \Bigl(\bigl\{\, x \in V : x \text{ is a leaf of } G \,\bigr\}
  \setminus \{\, z_C : C \text{ latent twin-cycle}\,\}\Bigr)
  \\
  &\;\cup\;
  \Bigl( \; \bigcup_{C \text{ unit}} \{u_C\} \;\Bigr)
  \;\cup\;
  \Bigl( \; \bigcup_{C \text{ twin or latent twin}} \{a_C, d_C\} \;\Bigr).
\end{align*}

A boundary vertex $v$ is called \emph{$\Sstar$-active} if
$\Sstar \cap V(G_v) \neq \emptyset$, and \emph{$\Sstar$-inactive}
otherwise.
\end{definition}

\begin{lemma}
\label{lem:size}
$\Sstar \cap \calA(G) = \emptyset$ and $|\Sstar| = \capa(G)$.
Moreover, a boundary vertex $v$ is $\Sstar$-inactive if and only if
$v = v^*_C$ for some latent twin-cycle $C$.
\end{lemma}

\begin{proof}
Leaves of $G$ are not articulation points; each $u_C$ lies on an
$\emptyset$-free run and hence outside $A_C$; and $a_C,d_C$ are not
articulation points by Lemma~\ref{lem:unique-run}.
Hence $\Sstar\cap\calA(G)=\emptyset$.
The three sets forming $\Sstar$ are pairwise disjoint: $u_C,a_C,d_C$
lie on a cycle and are not articulation points, so they have degree
two in $G$ and are not leaves, and each of them lies on exactly one
cycle component, so no vertex is selected by two cycle components.
For the cardinality, the map $C \mapsto z_C$ is injective, since
$z_C$ determines the beard containing it, which by (B2) is attached to
exactly one cycle component, namely $C$.
Therefore exactly $n_{\mathrm{lt}}$ leaves are removed and
\[
  |\Sstar| = \bigl(\ell(G)-n_{\mathrm{lt}}\bigr)
             + n_1 + 2n_2 + 2n_{\mathrm{lt}}
           = 2n_2+n_{\mathrm{lt}}+n_1+\ell(G) = \capa(G)
\]
by~\eqref{eq:capG}.

For the last assertion, if $v = v^*_C$ then $G_v$ is a beard whose
only non-articulation-point vertex is $z_C \notin \Sstar$, so
$\Sstar\cap V(G_v)=\emptyset$.
Conversely, let $v$ be a boundary vertex of a cycle component $C$ with
$\Sstar\cap V(G_v)=\emptyset$.
We first observe that no cycle component $C'$ of positive capacity
lies below $(C,v)$: otherwise the vertices selected on $C'$
($u_{C'}$, or $a_{C'}$ and $d_{C'}$) belong to $\Sstar$ and are not
articulation points, hence differ from $v$ and lie in $V(G_v)$.
By Lemma~\ref{lem:branch-value}, $\ccaps(C,v)\ge1$, so $G_v$ therefore
contains a leaf $x$ of $G$; as $x \notin \Sstar$, we have $x = z_{C'}$
for some latent twin-cycle $C'$, and $B := G_{v^*_{C'}}$ is the beard
containing $x$.
Since $\capa(C')=1>0$, the cycle $C'$ does not lie below $(C,v)$, \ie
$V(C')\not\subseteq V(G_v)\cup\{v\}$; as $C'\setminus\{v\}$ is
connected, this forces $V(C')\cap V(G_v)=\emptyset$ and in particular
$v^*_{C'}\notin V(G_v)$.
Now $B$ is a path with endpoint $v^*_{C'}$, so $B\setminus
\{v^*_{C'}\}$ is connected and contains $x \in V(G_v)$, hence is
contained in $V(G_v)$ by~\eqref{eq:branch-partition}; since
$v^*_{C'}$ is adjacent to it and no vertex of $G_v$ has a neighbour in
$H_v$, we conclude $v^*_{C'} = v$.
Finally $v$ lies on no cycle other than $C$ --- a second cycle through
$v$ would lie below $(C,v)$ and have positive capacity by
Lemma~\ref{lem:branch-value} (case $k'=1$ or $k'=2$ of its proof) ---
so $C' = C$ and $v = v^*_C$.
\end{proof}

\begin{lemma}[Local traversability]
\label{lem:local}
Let $C$ be a cycle component of $G$ and let
$X = \Sstar\cap V(C)$.
\begin{enumerate}
  \item[(T1)] For any two $\Sstar$-active boundary vertices
        $w,w'$ of $C$, some shortest $w$--$w'$ arc of $C$ avoids $X$.
  \item[(T2)] For every $x \in X$ and every $\Sstar$-active boundary
        vertex $w$ of $C$, some shortest $x$--$w$ arc of $C$ avoids
        $X \setminus \{x\}$ in its interior.
  \item[(T3)] If $|X| = 2$, the two vertices of $X$ are mutually
        visible with respect to $\Sstar$.
\end{enumerate}
\end{lemma}

\begin{proof}
By Definition~\ref{def:canonical}, $|X| \le 2$, and $|X|=2$ occurs
exactly when $C$ is a twin- or a latent twin-cycle, in which case
$X=\{a_C,d_C\}$ consists of the two endpoints of $P_C$.

(T3) is immediate: no vertex of $\Sstar$ outside $C$ lies on an arc of
$C$, so both arcs joining the two vertices of $X$ have interiors
disjoint from $\Sstar$, and the shorter of them is a witness.

If $|X| \le 1$, then (T2) is vacuous as well, because
$X\setminus\{x\}=\emptyset$ means that \emph{every} $x$--$w$ arc
qualifies.
If $C$ is a zero-cycle, (T1) is vacuous too since $X=\emptyset$.

Let $C$ be a unit-cycle, so $X = \{u_C\}$ with $u_C \in P_C$ and
$|P_C| = \tau_C$.
For (T1), boundary vertices are articulation points and hence lie
outside $P_C$, so $P_C$ is contained in one of the two $w$--$w'$ arcs;
that arc has length at least $|P_C|+1 = \tau_C+1 \geq |C|/2$
(using $|C| \leq 2\tau_C+2$), so the other arc, which avoids
$P_C \ni u_C$, is a shortest one.

Finally let $C$ be a twin-cycle or a latent twin-cycle, put
$X_C = \emptyset$ resp.\ $X_C = \{v^*_C\}$, and let
$\lambda = \lambda_C^{X_C} > \tau_C$ be the length of $P_C$, with
$X = \{a_C,d_C\}$ its endpoints.
By Lemma~\ref{lem:size}, every $\Sstar$-active boundary vertex of
$C$ lies outside $P_C$ (the only boundary vertex that may lie on $P_C$
is $v^*_C$, which is $\Sstar$-inactive).
Let $p$ (resp.\ $q$) be the vertex adjacent to $a_C$ (resp.\ $d_C$)
outside $P_C$.
Note that $\lambda \geq \tau_C+1 \geq |C|/2$.

For (T1), the arc containing $P_C$ has length at least
$\lambda+1 > |C| - \lambda -1$, hence the complementary arc, which
avoids $P_C \supseteq X$, is strictly shorter.

For (T2) with $x = a_C$, the two $a_C$--$w$ arcs have lengths
$1+\delta$ and $\lambda + \bigl(|C|-\lambda-1-\delta\bigr) =
|C|-1-\delta$, where $\delta = d_{C\setminus P_C}(p,w)$.
Since $\delta \le |C|-\lambda-1 \le |C|/2 - 1$, we get
$1+\delta \le |C|-1-\delta$, so the first arc is a shortest one; its
interior lies in $C \setminus P_C$ and therefore avoids $d_C$.
The case $x = d_C$ is symmetric.
\end{proof}

\begin{proposition}
\label{prop:construction}
$\Sstar$ is an \MVS\ of $G$, and consequently
$\mu(G) \geq \capa(G) = 2n_2+n_{\mathrm{lt}}+n_1+\ell(G)$.
\end{proposition}

\begin{proof}
Let $s,t \in \Sstar$, $s \neq t$.
Since every boundary vertex is a cut vertex and no vertex of $\Sstar$
is an articulation point (Lemma~\ref{lem:size}), the sequence of
components of the cycle--tree decomposition visited by an $s$--$t$ path
is the same for every such path, and every $s$--$t$ path is the
concatenation of sub-paths joining consecutive boundary vertices inside
single components.
Consequently a concatenation of shortest sub-paths is a shortest
$s$--$t$ path, and it remains to choose each sub-path $\Sstar$-free,
which is what Lemma~\ref{lem:local} provides.
Every boundary vertex $v$ visited in this way is $\Sstar$-active,
because $s$ or $t$ lies in $G_v$.
\begin{itemize}
  \item Inside a tree component the sub-path is unique, and all its
        internal vertices are articulation points and hence not in
        $\Sstar$ (Observation~\ref{obs:tree}).
  \item Inside an intermediate cycle component $C$ the sub-path joins
        two $\Sstar$-active boundary vertices of $C$; apply (T1).
  \item Inside the cycle component containing $s$ (or $t$), the
        sub-path joins $s$ (or $t$) to an $\Sstar$-active boundary
        vertex; apply (T2).
  \item If $s$ and $t$ lie in the same component, apply (T3) for a
        cycle component; for a tree component, both $s$ and $t$ are
        leaves of $G$ by Observation~\ref{obs:tree} and the internal
        vertices of the connecting path avoid $\Sstar$ by the same
        observation.
\end{itemize}
Hence $s$ and $t$ are mutually visible with respect to $\Sstar$, and
$\mu(G) \geq |\Sstar| = \capa(G)$ by Lemma~\ref{lem:size}
and~\eqref{eq:capG}.
\end{proof}

\subsection{The Upper Bound and the Exact Formula}
\label{sec:upper}

Throughout this subsection $S$ denotes an arbitrary \MVS\ of $G$ with
$S\cap\calA(G)=\emptyset$.
For a cycle component $C$ we call $v \in A_C$ \emph{$S$-active} if
$S\cap V(G_v)\neq\emptyset$ and \emph{$S$-inactive} otherwise, and we
write $I_C \subseteq A_C$ for the set of $S$-inactive boundary
vertices of $C$; thus $\lambda_C^{I_C}$ is the length of the longest
run of $C$ free of $S$-active boundary vertices.

\begin{lemma}
\label{lem:cycle-bound}
Let $C$ be a cycle component of $G$. Then $|S\cap V(C)| \leq 3$, and
\emph{(i)} $|S\cap V(C)| = 0$ if $\lambda_C^{I_C} < \tau_C$;
\emph{(ii)} $|S\cap V(C)| \leq 1$ if $\lambda_C^{I_C} = \tau_C$;
\emph{(iii)} $|S\cap V(C)| \leq 2$ if $\lambda_C^{I_C} > \tau_C$ and
$I_C \neq A_C$;
\emph{(iv)} if $|S\cap V(C)| = 3$ then $I_C = A_C$, \ie every boundary
vertex of $C$ is $S$-inactive.
\end{lemma}

\begin{proof}
By Observation~\ref{obs:isometric}, $S \cap V(C)$ is an \MVS\ of the
cycle $C$, so $|S\cap V(C)|\leq 3$ by Lemma~\ref{lem:cycle-mvn}.

If $I_C = A_C$ then $\lambda_C^{I_C} = |C| > \tau_C$ and only
statement (iv) is applicable, so assume $A_C\setminus I_C\neq\emptyset$
for statements (i)--(iii).
Let $u \in S\cap V(C)$; since $u\notin\calA(G)$, the vertex $u$ lies on
some $I_C$-free run $P$, delimited by two $S$-active boundary vertices
$p$ and $q$ (with $p=q$ and $|P|=|C|-1>\tau_C$ if
$|A_C \setminus I_C| = 1$).
If $|P| < \tau_C$, then by Lemma~\ref{lem:arc} the arc through $P$ is
the unique shortest $p$--$q$ path; picking $s_p \in S\cap V(G_p)$ and
$s_q \in S\cap V(G_q)$ (which exist because $p,q$ are $S$-active),
every shortest $s_p$--$s_q$ path passes through $p$, the arc, and $q$,
hence through $u \in S$ --- contradicting mutual visibility of
$s_p$ and $s_q$.
Thus every vertex of $S\cap V(C)$ lies on an $I_C$-free run of length
at least $\tau_C$; in particular (i) holds.

Assume now $\lambda_C^{I_C} = \tau_C$.
Suppose two vertices $u \neq u'$ of $S\cap V(C)$ lay on two distinct
$I_C$-free runs $P \neq P'$, necessarily both of length $\tau_C$.
Then $|C| \geq |P|+|P'|+|A_C\setminus I_C| \ge 2\tau_C+2$, and since
$|C| \le 2\tau_C+2$ we get $|C| = 2\tau_C+2$ and
$A_C \setminus I_C = \{p,q\}$ with $P,P'$ the two runs between $p$ and
$q$; the two $p$--$q$ arcs both have length $\tau_C+1$ and each
contains a vertex of $S$ in its interior, so $s_p$ and $s_q$ are not
mutually visible --- a contradiction.
Hence all vertices of $S\cap V(C)$ lie on one run $P$ with
$|P| = \tau_C$, delimited by $S$-active $p,q$.
If $u,u' \in P \cap S$ with $u$ closer to $p$, then by
Lemma~\ref{lem:arc} the sub-path of $P$ from $p$ to $u'$ is the unique
shortest $p$--$u'$ path and it passes through $u$, so $s_p$ and $u'$
are not mutually visible.
Thus $|S\cap V(C)|\leq 1$, proving (ii).

For (iv), suppose $|S\cap V(C)| = 3$, say $S\cap V(C)=\{u_1,u_2,u_3\}$,
which splits $C$ into three arcs.
Let $v \in A_C$; then $v \notin S$, so $v$ lies in the interior of one
of the arcs, say the one joining $u_1$ and $u_2$.
Any shortest path from a vertex of $G_v$ to $u_3$ leaves $G_v$ through
$v$ and then follows one of the two $v$--$u_3$ arcs of $C$, which pass
through $u_1$ resp.\ $u_2$; both are blocked.
Hence $S \cap V(G_v) = \emptyset$, \ie $v$ is $S$-inactive.
As $v \in A_C$ was arbitrary, $I_C = A_C$.
Statement (iii) is the contrapositive of (iv) combined with
$|S\cap V(C)|\le 3$.
\end{proof}

\begin{lemma}[Local exchange inequality]
\label{lem:exchange}
For every cycle component $C$ of $G$,
$
   |S \cap V(C)| \leq \capa(C) + \sum_{v \in I_C} \ccaps(C,v).
$
\end{lemma}

\begin{proof}
If $I_C = \emptyset$, then $\lambda_C^{I_C} = \lambda_C$ and
Lemma~\ref{lem:cycle-bound} gives $|S\cap V(C)| \leq 2$ if $C$ is a
twin-cycle, $\leq 1$ if $C$ is a unit-cycle, and $=0$ otherwise, which
is at most $\capa(C)$ in all four cases (for a latent twin-cycle
$\lambda_C<\tau_C$, so the left-hand side is $0 \le 1 = \capa(C)$).
 
Let $I_C \neq \emptyset$ and put $\kappa = \sum_{v\in I_C}\ccaps(C,v)$.
By Lemma~\ref{lem:branch-value}, $\kappa \geq |I_C|$, and
$\kappa \geq |I_C| + 1$ if some branch $G_v$ with $v \in I_C$ is not a
beard.
\begin{itemize}
  \item If $|S\cap V(C)| = 3$, then $I_C = A_C$ by
        Lemma~\ref{lem:cycle-bound}(iv).
        If $|A_C| = 1$, then $\lambda_C = |C|-1 > \tau_C$, so
        $\capa(C) = 2$, while $\kappa \geq 1$ by
        Lemma~\ref{lem:branch-value}; hence $\capa(C)+\kappa \geq 3$.
        If $|A_C| = 2$, then $\lambda_C \geq \tau_C$ by the
        computation in the proof of Lemma~\ref{lem:branch-value},
        so $\capa(C)\geq 1$ and $\capa(C)+\kappa\geq 1+2 = 3$.
        If $|A_C|\geq 3$, then $\kappa \geq 3$.
  \item If $|S\cap V(C)| = 2$, then $\lambda_C^{I_C} > \tau_C$ by
        Lemma~\ref{lem:cycle-bound}(i)--(ii).
        If $|I_C| \geq 2$, then $\kappa \geq 2$.
        If $I_C = \{v\}$ and $G_v$ is not a beard, then
        $\kappa \geq 2$.
        If $I_C = \{v\}$ and $G_v$ is a beard, then $v$ is a beard
        boundary vertex with $\lambda_C^{v} > \tau_C$; hence $C$ is a
        twin-cycle, a unit-cycle, or --- if $\lambda_C<\tau_C$ --- a
        latent twin-cycle, so $\capa(C)\geq 1$ and
        $\capa(C)+\kappa \geq 1+1 = 2$.
  \item If $|S\cap V(C)| \leq 1$, the inequality holds because
        $\kappa \geq |I_C| \geq 1$.
\end{itemize}
\end{proof}

The last ingredient is a bookkeeping statement which guarantees that,
when the local inequality of Lemma~\ref{lem:exchange} is summed over
all cycle components, no leaf and no cycle component is counted twice.

\begin{lemma}
\label{lem:bookkeeping}
Let $S \neq \emptyset$ be an \MVS\ of $G$ with
$S\cap\calA(G)=\emptyset$ and let $\mathcal{W}$ be the set of
$\subseteq$-maximal members of
$\{G_v : v \text{ is } S\text{-inactive}\}$.
Then:
\begin{enumerate}
 \item[(i)] the family $\{G_v : v \text{ is } S\text{-inactive}\}$ is
       laminar; consequently the members of $\mathcal{W}$ are pairwise
       disjoint, and each $H\in\mathcal{W}$ equals $G_{v_H}$ for a
       unique boundary vertex $v_H$;
 \item[(ii)] calling a cycle component $C'$ \emph{below} $H$ when
       $V(C')\subseteq V(H)\cup\{v_H\}$, no cycle component is below
       two distinct members of $\mathcal{W}$, and every cycle
       component below some $H\in\mathcal{W}$ satisfies
       $S\cap V(C')=\emptyset$;
 \item[(iii)] if a cycle component $C$ is below no member of $\mathcal{W}$
       and $v \in I_C$, then $G_v \in \mathcal{W}$, and the map
       $(C,v)\mapsto G_v$ is injective on such pairs;
 \item[(iv)] for $H = G_v$ as in \emph{(iii)},
       $\ccaps(C,v) = \ell_G(H) +
        \sum_{C' \text{ below } H}\capa(C')$.
\end{enumerate}
\end{lemma}

\begin{proof}
\emph{(i) Laminarity.}
Let $G_v$ (a branch of $C$) and $G_{v'}$ (a branch of $C'$) be
$S$-inactive.
We show that if they are neither disjoint nor nested, then
$V(G_v)\cup V(G_{v'}) \supseteq V \setminus \{z\}$ for some
$z \in \calA(G)$; since $S$ avoids both branches and $\calA(G)$, this
forces $S = \emptyset$, a contradiction.
If $v = v'$ (so $C \neq C'$), let $K_1,\dots,K_m$ be the components of
$G-v$ with $C\setminus\{v\}\subseteq K_1$ and
$C'\setminus\{v\}\subseteq K_2$; then $G_v = \bigcup_{i\neq 1}K_i$ and
$G_{v'} = \bigcup_{i\neq 2}K_i$, whose union is $V\setminus\{v\}$, and
$v \in \calA(G)$.
Let now $v \neq v'$.
If $v' \notin V(G_v)$ and $v \notin V(G_{v'})$, then
$v' \in V(H_v)$ by~\eqref{eq:branch-partition}, so $G_v \cup \{v\}$ is
connected in $G - v'$ and hence contained in the component of $G-v'$
containing $v$; that component is disjoint from $G_{v'}$, so
$G_v \cap G_{v'} = \emptyset$.
Otherwise, say $v' \in V(G_v)$, and let $K$ be the component of $G-v$
containing $v'$, so $K \subseteq G_v$.
Every component of $G-v'$ either contains $v$ or is contained in $K$.
If $G_{v'}$ consists only of components of the latter kind, then
$G_{v'} \subseteq K \subseteq G_v$ and the two branches are nested.
Otherwise $G_{v'}$ contains the component $L \ni v$ of $G-v'$, and
$L \supseteq V \setminus (K \cup \{v'\})$; together with
$G_v \supseteq K$ this gives
$V(G_v)\cup V(G_{v'}) \supseteq V\setminus\{v'\}$ with
$v' \in \calA(G)$.
Consequently the maximal members of a laminar family are pairwise
disjoint.
Finally, each $H \in \mathcal{W}$ determines its boundary vertex: if
$v \neq v'$ satisfied $G_v = G_{v'} = H$, then every component $K$ of
$H$ would be simultaneously a component of $G-v$ and of $G-v'$, so its
only neighbour outside $K$ would be both $v$ and $v'$, which is
impossible; and if $v = v'$ but the two branches came from different
cycles $C \neq C'$, then the branch of $C$ at $v$ would contain
$C'\setminus\{v\}$ while the branch of $C'$ at $v$ would not.

\emph{(ii)}
If a cycle component $C'$ were below two distinct
$H,H'\in\mathcal{W}$, then, since $V(H)\cap V(H')=\emptyset$,
$v_H\notin V(H)$ and $v_{H'}\notin V(H')$ by
\eqref{eq:branch-partition}, we would get
$V(C') \subseteq \bigl(V(H)\cup\{v_H\}\bigr)\cap
\bigl(V(H')\cup\{v_{H'}\}\bigr) \subseteq \{v_H,v_{H'}\}$,
contradicting $|V(C')|\ge3$.
If $C'$ is below $H$, then $V(C')\subseteq V(H)\cup\{v_H\}$, and $S$
avoids $V(H)$ as well as $v_H \in \calA(G)$, so
$S\cap V(C')=\emptyset$.

\emph{(iii)}
Let $C$ be below no member of $\mathcal{W}$ and $v \in I_C$.
By maximality, $G_v \subseteq H$ for some $H\in\mathcal{W}$.
Suppose the inclusion is strict.
By the case analysis in (i), $v_H \notin V(G_v)$ and
$v \notin V(H)$ would give $G_v \cap H = \emptyset$, which is
impossible as $G_v \neq \emptyset$; and $v_H \in V(G_v)$ would give
$H \subseteq G_v$ or $S = \emptyset$.
Hence $v \in V(H)$, and $v \neq v_H$.
The set $C\setminus\{v_H\}$ is connected (it is $C$ itself or a path)
and contains $v \in V(H)$, so it is contained in the component of
$G-v_H$ containing $v$, which is part of $H$; therefore
$V(C)\subseteq V(H)\cup\{v_H\}$, \ie $C$ is below $H$ ---
a contradiction.
Thus $G_v \in \mathcal{W}$.
For injectivity, if $v = v'$ and $C \neq C'$ are two such cycle
components with $v \in I_C \cap I_{C'}$, then $C'\setminus\{v\}$ is
contained in the branch of $C$ at $v$, so $C'$ is below
$G_v \in \mathcal{W}$, which is excluded; and if $v \neq v'$ then
$G_v \neq G_{v'}$ by~(i).

\emph{(iv)}
By (i) we have $v = v_H$ for $H = G_v$, so
$\calC(C,v) = \{C' : V(C')\subseteq V(H)\cup\{v_H\}\}$ is exactly the
set of cycle components below $H$, and the claim is
Definition~\ref{def:branchcap}.
\end{proof}

\begin{theorem}
\label{thm:mvn}
Let $G$ be a cactus graph that is not a simple cycle. Then
$
   \mu(G) = \capa(G) = 2n_2 + n_{\mathrm{lt}} + n_1 + \ell(G).
$
\end{theorem}

\begin{proof}
The inequality $\mu(G)\ge\capa(G)$ is
Proposition~\ref{prop:construction}.
For the converse, let $S$ be a maximum \MVS\ with
$S\cap\calA(G)=\emptyset$, which exists by Lemma~\ref{lem:no-ap} and
is non-empty.
Every vertex of $V\setminus\calA(G)$ lies in exactly one component of
the cycle--tree decomposition, so
$|S| = \sum_{C}|S\cap V(C)| + \sum_{T}|S\cap V(T)|$, where $C$ ranges
over cycle components and $T$ over tree components; by
Observation~\ref{obs:tree}, $S\cap V(T)$ consists of leaves of $G$.
Let $\mathcal{W}$ be as in Lemma~\ref{lem:bookkeeping}.
As $S$ avoids $\bigcup\mathcal{W}$ and the members of $\mathcal{W}$
are pairwise disjoint,
\[
  \sum_{T}|S\cap V(T)| \;\le\; \ell(G) -
  \!\!\sum_{H\in\mathcal{W}}\!\!\ell_G(H),
  \qquad
  \sum_{C \text{ below some } H}\!\!\! |S\cap V(C)| \;=\; 0 .
\]
By Lemma~\ref{lem:bookkeeping}(iii)--(iv), summing
Lemma~\ref{lem:exchange} over the cycle components that are below no
member of $\mathcal{W}$ contributes each $H\in\mathcal{W}$ at most
once, with
$\ccaps(C,v)=\ell_G(H)+\sum_{C'\text{ below }H}\capa(C')$.
Hence
\begin{align*}
  |S| &\leq \sum_{C \text{ below no } H}
            \Bigl(\capa(C) + \sum_{v\in I_C}\ccaps(C,v)\Bigr)
        \;+\; \ell(G) - \sum_{H\in\mathcal{W}}\ell_G(H)\\[1pt]
      &\leq \sum_{C \text{ below no } H} \capa(C)
        \;+\; \sum_{H\in\mathcal{W}}
              \Bigl(\ell_G(H) + \sum_{C' \text{ below } H}\capa(C')
              \Bigr)
        \;+\; \ell(G) - \sum_{H\in\mathcal{W}}\ell_G(H)\\[1pt]
      &= \sum_{C\in\calC(G)} \capa(C) + \ell(G)
       \;=\; \capa(G) \;=\; 2n_2+n_{\mathrm{lt}}+n_1+\ell(G).
       \qedhere
\end{align*}
\end{proof}

Figure \ref{Fig:cactusmaxmvs} shows a maximum mutual visibility set of the given cactus graph,
$\mu(G) = 24$.

\begin{figure}[tb]
 \begin{center}
\includegraphics[width=\linewidth]{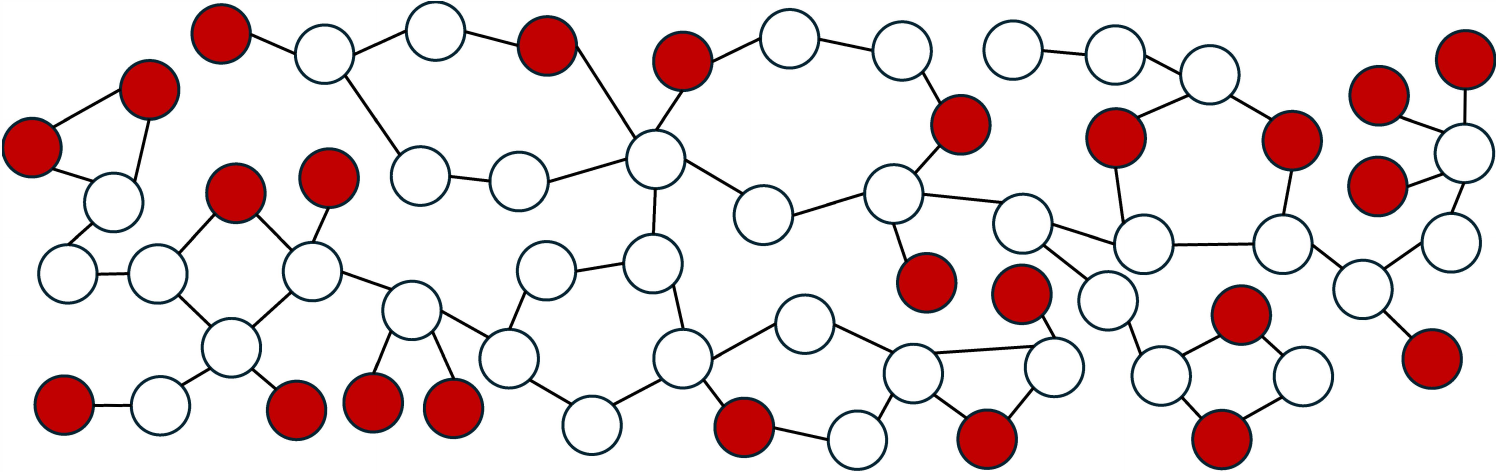}
  \caption{A maximum mutual visibility set of the given cactus graph ($\mu(G) = 24$)}
  \label{Fig:cactusmaxmvs}
 \end{center}
\end{figure}

\begin{corollary}
\label{cor:upper-bound}
For every cactus graph $G$ that is not a simple
cycle, $\mu(G) \leq 2|\calC(G)| + \ell(G)$, and equality holds if and
only if every cycle component of $G$ is a twin-cycle.
\end{corollary}

\begin{proof}
By Definition~\ref{def:capacity} we have $\capa(C) \leq 2$ for every
cycle component, with equality if and only if $C$ is a twin-cycle;
now apply Theorem~\ref{thm:mvn} and \eqref{eq:capG}.
\end{proof}

\begin{corollary}
\label{cor:linear}
A maximum \MVS\ of a cactus graph $G$ on $n$ vertices, and hence
$\mu(G)$, can be computed sequentially in $O(n)$ time.
\end{corollary}

\begin{proof}
The blocks and the articulation points of $G$ are computed by one
depth-first search in $O(n+|E|)=O(n)$ time, since $|E|\le
\frac{3}{2}(n-1)$ in a cactus; this yields the cycle--tree
decomposition of Definition~\ref{def:decomposition}.
Traversing each cycle component $C$ once gives $|C|$, $\tau_C$, all
maximal $\emptyset$-free runs, $\lambda_C$ and, for each boundary
vertex $v$, the value $\lambda_C^{v}=\max\{\lambda_C,
\alpha_v+1+\beta_v\}$; the total cost is
$O(\sum_{C}|C|)=O(n)$ by the bound $\sum_C |C|\le n+|\calC(G)|-1$.
One traversal of the tree components identifies the beards and their
boundary vertices, after which
Definition~\ref{def:cycle-types} classifies every cycle component and
Definition~\ref{def:canonical} selects the vertices of $\Sstar$, both
in $O(n)$ time.
Correctness is Theorem~\ref{thm:mvn} and
Proposition~\ref{prop:construction}.
\end{proof}

\begin{remark}
\label{rem:blockgraphs}
Cactus graphs and block graphs (every block a clique) properly
overlap without either containing the other, and the \MVN\ of block
graphs is already known~\cite{DiStefano2022}; the two analyses are
nevertheless of a different nature.
In a block graph every block has diameter one, so any two vertices of
a block are mutually visible regardless of the rest of the set and the
analysis is governed by the cut-vertex structure alone.
In a cactus a cycle block has diameter $\lfloor|C|/2\rfloor$, so
distances \emph{inside} a block matter and a chosen vertex may destroy
the visibility of two vertices of the same block---which is precisely
what the comparison of $\lambda_C$ with $\tau_C$ measures, and which
has no counterpart for cliques, as has the latent twin-cycle.
\end{remark}

\begin{remark}
\label{rem:notlocal}
Lemma~\ref{lem:cycle-bound} cannot be strengthened to
``$|S\cap V(C)|=\capa(C)$ for every maximum \MVS\ $S$''.
For the $5$-cycle $p,w_1,w_2,q,r$ whose only articulation points are
$p$ and $q$, each carrying a pendant leaf $a$ resp.\ $b$, we have
$\lambda_C=\tau_C=2$ and $\mu(G)=3$, yet both $\{a,b,w_1\}$ and
$\{a,w_1,r\}$ are maximum \MVS s and the latter takes two vertices of
$C$ while omitting $b$.
This is why we combine an upper bound (Lemma~\ref{lem:exchange}) with
an explicit construction (Definition~\ref{def:canonical}).
\end{remark}

\section{Self-Stabilizing Algorithms for a Maximum Mutual Visibility Set}
\label{sec:algorithm}
By Theorem~\ref{thm:mvn}, a maximum \MVS\ is determined by the
cycle--tree decomposition of $G$ together with the type of each cycle
component, and Definition~\ref{def:canonical} turns this
classification into an \emph{explicit} maximum \MVS\ $\Sstar$.
Both algorithms below compute exactly $\Sstar$, so their correctness
follows from Theorem~\ref{thm:mvn} and
Proposition~\ref{prop:construction} once they are shown to evaluate
Definition~\ref{def:canonical} correctly.
Only two of its choices are not already unique, and both are resolved
with process identifiers: for a \emph{unit-cycle} there may be up to
two $\emptyset$-free runs of length $\lambda_C$, exactly one of which
hosts the single \MVS\ vertex; and a \emph{latent twin-cycle} may have
several witnesses, exactly \emph{one} of which may be deactivated,
since deactivating a second one would forfeit a further leaf without
any gain.
For a \emph{twin-cycle} no symmetry breaking is needed, the longest
$\emptyset$-free run being unique (Lemma~\ref{lem:unique-run}).
The two algorithms share this correctness argument and differ only in
their optimization criterion: Algorithm~1 (Section~\ref{sec:alg1})
uses a single BFS tree and is space-efficient, 
whereas Algorithm~2 (Section~\ref{sec:alg2}) uses parallel BFS trees
from all high-degree vertices and stabilizes in time proportional to
the largest component rather than to the diameter.

Section~\ref{sec:model} fixes the computational model and the
specification of the task.
Section~\ref{sec:alg1} describes the strategy of Algorithm~1 in four
phases, Section~\ref{sec:rules} turns that strategy into a complete
set of rules in the state-reading model, and
Section~\ref{sec:alg1-analysis} proves its correctness and its
complexity.
Section~\ref{sec:alg2} obtains Algorithm~2 from Algorithm~1 by
replacing the single BFS tree by a family of local ones, and proves
the matching lower bound; Section~\ref{sec:comparison} compares the
two algorithms.

\subsection{System Model and Specification}
\label{sec:model}
\subparagraph*{Network and communication.}
The system is modeled by a cactus graph $G=(V,E)$ that is neither a
simple cycle nor a simple path (for those topologies a maximum \MVS\
is trivial), where each vertex is a \emph{process} and each edge a
bidirectional \emph{communication link}.
We adopt the \emph{state reading} (shared memory) model: each process
$v$ reads its own variables and those of all $u\in N(v)$, and is a
state machine whose next state is a function of these.
Each process $v$ has a distinct identifier $\mathrm{id}(v)$ taken from
a domain of size $n^{O(1)}$, and every process knows a common upper
bound $N\ge n$ on the number of processes; all distance variables
range over the finite domain $\{0,1,\dots,N\}$, and a value read
outside this domain is treated as $N$.
Both assumptions are standard and necessary here: with an unbounded
distance domain the space complexity could not be bounded by any
function of $n$, and with arbitrary integers as initial values the
stabilization time would depend on the magnitude of the corruption
rather than on the topology.
There is one predetermined leader $r$ with $\delta(r)\ge3$; such a
vertex always exists in our topologies.

\subparagraph*{Executions and the daemon.}
A \emph{configuration} $\gamma=(s_v)_{v\in V}$ is the tuple of all
process states.
Each process has a finite set of \emph{rules}, each of the form
$y_v \leftarrow F(\cdot)$ where $F$ depends only on the state of $v$
and of its neighbours; a process is \emph{enabled} in $\gamma$ if the
value of one of its variables differs from the value prescribed by the
corresponding rule, and executing a \emph{step} means that all rules
of the process are applied simultaneously.
We assume the \emph{distributed daemon}: at each step an arbitrary
non-empty subset of the enabled processes executes simultaneously.
An \emph{execution} is a maximal sequence
$\gamma_0,\gamma_1,\dots$ of configurations linked by steps; no
assumption is placed on $\gamma_0$, which models arbitrary transient
faults (the code itself is incorruptible).
We further assume the daemon to be \emph{weakly fair}: a process that
remains continuously enabled eventually executes a step.
This assumption is what makes a round-based complexity measure
meaningful---an unfair daemon may starve one process forever, so no
algorithm admits a bound in rounds under it---and it is the standard
assumption in the literature whenever stabilization time is measured
in rounds.
The \emph{first round} of an execution is its shortest prefix in which
every process enabled in $\gamma_0$ either executes a step or becomes
disabled; the rounds of the remaining suffix are defined inductively.

\subparagraph*{Output, specification and legitimacy.}
Every process $v$ has a boolean \emph{output variable}
$\mathrm{inMVS}_v$, and the \emph{output} of a configuration $\gamma$
is $\mathcal{O}(\gamma)=\{v\in V \mid \mathrm{inMVS}_v=\textsf{true}\}$.
The \emph{maximum-\MVS\ construction task} requires that
$\mathcal{O}(\gamma)$ be a maximum \MVS\ of $G$ containing no
articulation point of $G$; by Lemma~\ref{lem:no-ap} this restriction
is without loss of generality, and it makes the specification a
function of $G$ alone.
A configuration is \emph{legitimate} if no process is enabled in it.
Thus legitimacy is a purely local, syntactic predicate; it is not
defined in terms of the output.
That the two coincide is a theorem, proved in
Section~\ref{sec:alg1-analysis}: the rules of Algorithm~1 admit
\emph{exactly one} legitimate configuration, and its output is the
maximum \MVS\ $\Sstar$ of Definition~\ref{def:canonical}
(Theorem~\ref{thm:B-main}).
An algorithm is \emph{self-stabilizing} for the task if every
execution contains a legitimate configuration, and \emph{silent} if no
variable changes once a legitimate configuration is reached; since a
legitimate configuration is by definition a fixed point of the rules,
both of our algorithms are silent, and the output is constant from
that point on.

\subsection{Algorithm 1: A Space-Efficient Self-Stabilizing Algorithm}
\label{sec:alg1}
Algorithm~1 evaluates Definition~\ref{def:canonical} using a
\emph{single-root BFS} tree, in four phases.

\smallskip\noindent
\textbf{Phase 1 (BFS tree).}
A single BFS tree $\calB_r$ rooted at the leader $r$ is constructed by
the classical distance-plus-one rule; self-stabilizing BFS
constructions with a designated root go back to Huang and
Chen~\cite{HuangChen1992} and to Dolev, Israeli and
Moran~\cite{DolevIsraeliMoran1993} (see
also~\cite[Ch.~2]{Dolev2000}).
So that the paper is self-contained and so that the round complexity
does not depend on the variant chosen, we state and prove the bound we
use in Lemma~\ref{lem:bfs-time} of Section~\ref{sec:rules}: with
the distance domain $\{0,\dots,N\}$ and the root pinned to $0$, the
construction stabilizes in $O(D)$ rounds under the distributed daemon.

\smallskip\noindent
\textbf{Phase 2 (cycle identification and articulation points).}
Because the cycles of a cactus are edge-disjoint, every cycle of
$\calC(G)$ contains exactly one non-tree edge of $\calB_r$ and,
conversely, the fundamental cycle of a non-tree edge is a cycle of
$\calC(G)$; and $\{u,v\}$ is a non-tree edge exactly when
$\mathtt{par}_u\neq v$ and $\mathtt{par}_v\neq u$, which both
endpoints check locally.
We therefore name the cycle containing $e_C=\{a,b\}$ by
$\mathrm{cid}(C)=(\min\{\mathrm{id}(a),\mathrm{id}(b)\},
\max\{\mathrm{id}(a),\mathrm{id}(b)\})$---an identifier that needs no
prior knowledge of $V(C)$---and establish its membership by one
leaf-to-root wave along $\calB_r$, at the end of which the vertex
$x_C=\mathrm{lca}(a,b)$ of minimum depth computes
$|C|=\mathtt{dist}_a+\mathtt{dist}_b-2\,\mathtt{dist}_{x_C}+1$ and
$\tau_C$ and broadcasts them back down; see
Section~\ref{sec:rules} for the rules and their analysis.
Articulation points need no DFS low-link computation: in a cactus a
vertex of degree at least three is always an articulation point, and a
vertex of degree two is one exactly when it lies on no cycle.

\smallskip\noindent
\textbf{Phase 3 ($\lambda_C$, $\lambda_C^{v}$ and beard detection).}
For each cycle $C$ the parameter $\lambda_C$ is computed by propagating
run-length information around $C$, and each boundary vertex $v$ of $C$
additionally computes the lengths $\alpha_v,\beta_v$ of the two
$\emptyset$-free runs adjacent to it and stores
$\lambda_C^{v}=\max\{\lambda_C,\ \alpha_v+1+\beta_v\}$, the value of
$\lambda_C$ obtained by deactivating $v$ alone.
Beard detection is a leaf-to-root propagation along $\mathcal{B}_r$:
each leaf sends upward the maximum number of children seen so far, and
a tree component $T$ is flagged as a beard iff that count never exceeds
one along the whole path to the boundary vertex (no branching in $T$)
and that boundary vertex belongs to exactly one cycle, as required
by~(B2).
A beard boundary vertex $v$ is deactivated \emph{only} when
$\lambda_C<\tau_C$ and $\lambda_C^{v}>\tau_C$, \ie only when $C$ is a
latent twin-cycle with witness $v$; deactivating a beard
unconditionally would forfeit its leaf without any gain, for instance
when $\lambda_C=\tau_C$ already.
Among several witnesses, only the one of smallest ID is deactivated.

\smallskip\noindent
\textbf{Phase 4 (classification and selection).}
Each cycle is classified from $\lambda_C$, $\lambda_C^{v}$ and $\tau_C$
(Definition~\ref{def:cycle-types}), and each process sets
$\mathrm{inMVS}_v$ exactly as prescribed by
Definition~\ref{def:canonical}: the two endpoints of the unique longest
run for a twin- or latent twin-cycle (Lemma~\ref{lem:unique-run}); one
vertex of one run of length $\tau_C$ for a unit-cycle, the tie between
the at most two such runs being broken by ID; nothing for a zero-cycle;
every leaf of $G$ except those whose beard boundary vertex was
deactivated in Phase~3.

\subsection{The Rules of Algorithm~1}
\label{sec:rules}

The description of Section~\ref{sec:alg1} is a description of a
\emph{strategy}: it says which quantity is computed in each phase, but
it leaves the phrases ``propagating run-length information around
$C$'' and ``a leaf-to-root propagation along $\calB_r$'' informal, and
in a self-stabilizing setting such phrases hide exactly the two
questions that matter, namely what a process does when the data it
reads are corrupted, and how many rounds an aggregation costs.
This section removes that informality: it gives the complete set of
rules of Algorithm~1, together with the analysis tool
(Lemma~\ref{lem:B-layer}) that converts each aggregation into an
explicit round bound.

We work in the state-reading model, in which a process may read only
its own variables and those of its neighbours.
Every rule is a \emph{total} assignment
$y_v \leftarrow F\bigl(\text{local topology of } v,\
\{y'_u\}_{u\in N(v)}\bigr)$ whose right-hand side depends only on
variables of $v$ and of its neighbours; a process is \emph{enabled}
whenever one of its variables differs from the value prescribed by its
rule, and it executes \emph{all} of its rules when activated.
No rule is conditional on the current value of the variable it
assigns, so an arbitrary initial value is always overwritten rather
than preserved; this, together with the finiteness of every variable
domain, is what purges corrupted states, and it is also why the
algorithm is silent.

\subsubsection*{A layered evaluation lemma}

All computations are organised into \emph{layers}.
A layer consists of one variable per process (or per cycle incidence,
see the list of variables below) together with an \emph{evaluation
order}: an acyclic orientation of a subgraph of $G$ such that,
\emph{within the layer}, the rule of a process depends only on its
in-neighbours.
The \emph{depth} of a layer is the length of the longest directed path
of its evaluation order.

The following lemma is stated in a per-variable form.
The reason is that the analysis of Algorithm~2
(Section~\ref{sec:alg2}) needs to conclude that \emph{some} variables
are stable after few rounds even though the system as a whole is not;
a lemma that only bounds the global stabilization time cannot give
that.

\begin{lemma}[Layered evaluation]
\label{lem:B-layer}
Let $y^1,\dots,y^k$ be layers such that the rule of $y^i_v$ depends
only on the local topology of $v$, on $y^j_u$ for $j<i$ and
$u\in N(v)\cup\{v\}$, and on $y^i_u$ for the in-neighbours $u$ of $v$
in the evaluation order of layer~$i$.
For a variable instance $y=y^i_v$ define its \emph{dependency depth}
$\Delta(y)$ to be $0$ if the rule of $y$ refers to no other variable
of the family, and $1+\max\Delta(y')$ over the variables $y'$ that the
rule of $y$ refers to otherwise; the recursion is well founded because
the dependency relation is acyclic.
Write $\pi(y)$ for the value that the rule of $y$ prescribes when all
the variables it refers to hold their own prescribed values.
Then, under the weakly fair distributed daemon and from every initial
configuration, every execution satisfies: for every variable instance
$y$, from the end of round $\Delta(y)+1$ onwards, $y=\pi(y)$.
In particular every variable holds its prescribed value after at most
$1+\sum_{i=1}^{k}(h_i+1)$ rounds, where $h_i$ is the depth of
layer~$i$, and no process is enabled from that point on.
\end{lemma}

\begin{proof}
Recall that a process is enabled in a configuration precisely when at
least one of its variables differs from the value its rule prescribes
\emph{in that configuration}, that a step recomputes all the variables
of the process at once, and that, by the definition of a round, a
process that is enabled at the beginning of a round either executes a
step during that round or becomes disabled during it.
We induct on $\Delta(y)$; let $p$ be the process owning $y$.

Let $\Delta(y)=0$.
The right-hand side of the rule of $y$ is a function of the local
topology of $p$ alone, hence is constant along the execution and equal
to $\pi(y)$.
During round~$1$, either $p$ executes a step --- after which
$y=\pi(y)$ --- or $p$ becomes disabled, which by definition means that
all of its variables, $y$ included, already agree with the prescribed
values.
Every later step of $p$ re-assigns the same constant, so $y=\pi(y)$
from the end of round~$1$ on.

Let $\Delta(y)=d>0$ and assume the claim for every variable of
dependency depth less than $d$.
Every variable that the rule of $y$ refers to has dependency depth at
most $d-1$ and therefore holds its prescribed value from the end of
round~$d$ on.
Consequently, from the end of round~$d$ onwards the right-hand side of
the rule of $y$ evaluates to $\pi(y)$ and does not change any more.
The argument of the base case, applied to round $d+1$, then gives
$y=\pi(y)$ from the end of round $d+1$ on.

For the global bound, a variable of layer~$i$ that sits at position
$h$ of the evaluation order of that layer has dependency depth at most
$\sum_{j<i}(h_j+1)+h$, whence
$\max_y\Delta(y)\le\sum_{i=1}^{k}(h_i+1)$.
Once every variable holds its prescribed value, no process is enabled
by definition.
\end{proof}

\begin{corollary}[Local stabilization]
\label{cor:B-local}
Under the hypotheses of Lemma~\ref{lem:B-layer}, a set $Y$ of variable
instances with $\max_{y\in Y}\Delta(y)=O(k)$ holds its prescribed
values after $O(k)$ rounds, \emph{whatever the dependency depths of
the remaining variables are}.
\end{corollary}

\noindent
Lemma~\ref{lem:B-layer} is the tool that replaces informal phrases
such as ``propagated within $C$'': each aggregation below is realised
as an explicit layer with an explicit evaluation order, an explicit
depth, and an explicit purge behaviour.

\subsubsection*{Notation and variables}

We write $N(v)$ for the neighbours of $v$, $\delta(v)=|N(v)|$,
$\mathrm{ch}(v)=\{u\in N(v)\mid \mathtt{par}_u=v\}$ for the children
of $v$ in $\calB_r$, and
\[
   E^{\mathrm{nt}}_v=\{u\in N(v)\mid
     \mathtt{par}_u\neq v \wedge \mathtt{par}_v\neq u\},
   \qquad
   \mathrm{key}(v,u)=\bigl(\min\{\mathrm{id}(v),\mathrm{id}(u)\},
     \max\{\mathrm{id}(v),\mathrm{id}(u)\}\bigr)
\]
for the non-tree neighbours of $v$ and the identifier of the cycle
closed by such an edge.
A process stores a constant part and one \emph{record} per incident
cycle; $\mathtt{Rec}_v$ is the set of identifiers for which $v$ holds
a record and $\mathtt{z}_v[c]$ is field $\mathtt{z}$ of record $c$.
This per-incidence representation is what allows a process on
$\gamma_v\ge2$ cycles---for instance the centre of a friendship
graph---to keep two cycle identifiers, two pairs of cycle neighbours
and two values of $\lambda_C$ at the same time.

\medskip
\noindent
\begin{tabular}{@{}p{0.40\textwidth}p{0.55\textwidth}@{}}
$\mathtt{par}_v,\ \mathtt{dist}_v$
  & parent and depth in $\calB_r$, $\mathtt{dist}_v\in\{0,\dots,N\}$\\
$\mathtt{isAP}_v,\ \mathtt{chain}_v,\ \mathtt{deAbv}_v,\
 \mathtt{inMVS}_v$
  & booleans\\
$\mathtt{Src}_v[c],\ \mathtt{src}_v[c],\ \mathtt{up}_v[c],\
 \mathtt{bot}_v[c]$
  & sources of $c$, their number, forwarding flag, depth of the
    endpoint below\\
$\mathtt{distx}_v[c],\ \mathtt{size}_v[c],\ \mathtt{tau}_v[c]$
  & $\mathrm{dist}_{x_C}$, $|C|$, $\tau_C$\\
$\mathtt{arm0}_v[c],\ \mathtt{arm}_v[c],\ \mathtt{pos}_v[c]$
  & source labelled arm~$0$, side of $C$, position along $C$\\
$\mathtt{lrun}_v[c],\ \mathtt{rrun}_v[c],\ \mathtt{run}_v[c]$
  & run lengths to the left, to the right, and in total\\
$\mathtt{pmax}_v[c],\ \mathtt{Lam}_v[c]$
  & running maximum of $\mathtt{run}$, and $\lambda_C$\\
$\mathtt{alpha}_v[c],\ \mathtt{beta}_v[c],\ \mathtt{lamDe}_v[c]$
  & the two adjacent runs, and $\lambda_C^{v}$\\
$\mathtt{bBnd}_v[c],\ \mathtt{wit}_v[c]$
  & beard boundary flag, witness flag\\
$\mathtt{wmin}_v[c],\ \mathtt{wpos}_v[c],\ \mathtt{deact}_v[c]$
  & witness selection along $C$\\
$\mathtt{lrun}'_v[c],\dots,\mathtt{Lam}'_v[c]$
  & the five run variables recomputed after deactivation\\
$\mathtt{fmin}_v[c],\ \mathtt{fpos}_v[c]$
  & first position at which a longest run starts\\
$\mathtt{cyctype}_v[c],\ \mathtt{sel}_v[c]$
  & type of $C$, and selection predicate\\
\end{tabular}

\medskip\noindent
Every field is an identifier pair, a value of $\{0,\dots,N\}\cup
\{\infty\}$, or a constant-size flag, so a record occupies
$O(\log n)$ bits.
In an arbitrary configuration $v$ holds at most
$|\mathtt{Cand}_v|\le\delta(v)$ records (the rule for $\mathtt{Cand}_v$ in
Algorithm~\ref{alg:B1}), hence
$O\bigl((\delta(v)+1)\log n\bigr)$ bits, and after stabilization
exactly $\gamma_v$ records, hence $O\bigl((\gamma_v+1)\log n\bigr)$
bits; since $\sum_{v}\delta(v)=2|E| = O(n)$ in a cactus, both bounds
give a total of $O(n\log n)$ bits and an average of $O(\log n)$ bits
per process, as established in Proposition~\ref{prop:alg1-space}.

\subsubsection*{Layers 1--2: BFS tree and cycle identification}

\begin{algorithm}[tb]
\caption{Rules for process $v$ --- Layer~1 (BFS tree) and Layer~2 (cycle identification)}
\label{alg:B1}
\begin{algorithmic}[1]
\Statex \textbf{Layer 1 (BFS tree; evaluation order: $\calB_r$ from
        the root, depth $\le D$)}
\If{$v=r$}
   \State $\mathtt{par}_v \gets \bot$;\quad $\mathtt{dist}_v \gets 0$
\Else
   \State $\mathtt{par}_v \gets
     \arg\min_{u\in N(v)}
     \bigl(\min\{\mathtt{dist}_u,N\},\ \mathrm{id}(u)\bigr)$
   \State $\mathtt{dist}_v \gets
     \min\bigl\{\min\{\mathtt{dist}_{\mathtt{par}_v},N\}+1,\ N\bigr\}$
     \Comment{finite domain: out-of-range values are read as $N$}
\EndIf
\Statex
\Statex \textbf{Layer 2 (cycle identification; evaluation order:
        $\calB_r$ from the leaves, depth $\le D$)}
\State $\mathtt{Cand}_v \gets
   \{\mathrm{key}(v,u) \mid u\in E^{\mathrm{nt}}_v\}\ \cup$
\Statex \hskip 1.2em
   $\bigcup_{u\in \mathrm{ch}(v)}
   \{c\in\mathtt{Rec}_u \mid \mathtt{up}_u[c]\}$
   \Comment{at most $\delta(v)$ candidates}
\ForAll{$c \in \mathtt{Cand}_v$}
   \State $\mathtt{Src}_v[c] \gets
     \bigl\{(\bot,\mathtt{dist}_v) \mid
       \exists u\in E^{\mathrm{nt}}_v:\mathrm{key}(v,u)=c \bigr\}$
   \Statex \hskip 2.2em $\uplus\
     \bigl\{(u,\mathtt{bot}_u[c]) \mid u\in\mathrm{ch}(v),\
       c\in\mathtt{Rec}_u,\ \mathtt{up}_u[c]\bigr\}$
\EndFor
\State $\mathtt{Rec}_v \gets
   \{\, c\in\mathtt{Cand}_v \mid |\mathtt{Src}_v[c]| \le 2 \,\}$
   \Comment{$|\mathtt{Src}_v[c]|\ge3$ is impossible in a cactus:
            purge}
\ForAll{$c \in \mathtt{Rec}_v$}
   \State $\mathtt{src}_v[c] \gets |\mathtt{Src}_v[c]|$;\quad
          $\mathtt{up}_v[c] \gets (\mathtt{src}_v[c]=1)$
   \State \textbf{if} $\mathtt{src}_v[c]=1$ \textbf{then}
          $\mathtt{bot}_v[c] \gets \beta$ where
          $\mathtt{Src}_v[c]=\{(\cdot,\beta)\}$
          \textbf{else} $\mathtt{bot}_v[c] \gets \bot$
\EndFor
\end{algorithmic}
\end{algorithm}

An edge $\{u,v\}$ is a non-tree edge of $\calB_r$ exactly when
$\mathtt{par}_u\neq v$ and $\mathtt{par}_v\neq u$, which both
endpoints check locally.
Because the cycles of a cactus are edge-disjoint, each cycle contains
exactly one non-tree edge and, conversely, the fundamental cycle of a
non-tree edge is a cycle of $\calC(G)$; naming $C$ after its non-tree
edge therefore requires no prior knowledge of $V(C)$ and removes the
circularity inherent in naming a cycle after its minimum-depth vertex.

\begin{lemma}[Pinned BFS]
\label{lem:bfs-time}
Consider the rule $\mathtt{dist}_\rho=0$ and, for $v\neq\rho$,
$\mathtt{dist}_v \leftarrow \min\{\min_{u\in N(v)}\mathtt{dist}_u+1,
N\}$ over the domain $\{0,\dots,N\}$, a value read outside the domain
counting as $N$.
From every initial configuration, and under the weakly fair
distributed daemon, the following holds for every $k\ge0$: at the end
of round $k+1$ we have $\mathtt{dist}_v=d_G(\rho,v)$ for every $v$
with $d_G(\rho,v)\le k$, and the parent pointers of those vertices are
correct one round later.
In particular all distances are correct after $D+1$ rounds; and the
distances inside the ball of radius $k$ around $\rho$ are correct
after $k+1$ rounds \emph{whatever the diameter of $G$ is}.
\end{lemma}

\begin{proof}
Write $d(v)=d_G(\rho,v)$.

\emph{Upper direction.}
We show by induction on $\rho'$ that, from the end of round $\rho'$
on, $\mathtt{dist}_v\le d(v)$ holds for every $v$ with
$d(v)\le\rho'-1$.
For $\rho'=1$ this concerns only $v=\rho$, whose rule pins
$\mathtt{dist}_\rho=0$; during round~$1$ the process $\rho$ either
executes or is disabled, so the value is $0$ at the end of round~$1$
and never changes again.
Assume the property after round $\rho'$ and let $d(v)=\rho'$.
The vertex $v$ has a neighbour $u$ with $d(u)=\rho'-1$, and
$\mathtt{dist}_u\le d(u)$ holds throughout round $\rho'+1$ by the
induction hypothesis and by the fact that the property is preserved by
every step: a process $w$ with $d(w)\le\rho'-1$ that executes sets
$\mathtt{dist}_w\le\mathtt{dist}_{u'}+1\le d(u')+1=d(w)$ for a
neighbour $u'$ on a shortest $\rho$--$w$ path.
Hence during round $\rho'+1$ the process $v$ either executes, and then
$\mathtt{dist}_v\le\mathtt{dist}_u+1\le d(v)$, or is disabled, in
which case its value already equals the one its rule prescribes, which
is at most $d(v)$ for the same reason.

\emph{Lower direction.}
The property ``$\mathtt{dist}_v\ge\min\{d(v),\rho'\}$ for every $v$''
is preserved by every step, because a process that executes sets
$\mathtt{dist}_v\ge\min\{\min_{u\in N(v)}\min\{d(u),\rho'\}+1,N\}
=\min\{d(v),\rho'+1\}\ge\min\{d(v),\rho'\}$, using
$\min_{u\in N(v)}d(u)=d(v)-1$ and $d(v)\le D\le N$.
We show by induction on $\rho'$ that it holds at the end of round
$\rho'$; it is vacuous for $\rho'=0$.
Assume it after round $\rho'$.
Every process either executes at some point of round $\rho'+1$, and
the displayed computation then gives
$\mathtt{dist}_v\ge\min\{d(v),\rho'+1\}$, or is disabled at some point
of that round, in which case its value equals the one prescribed by
its rule and the same computation applies; and the bound survives the
remaining steps of the round by the preservation property.

\emph{Conclusion.}
At the end of round $k+1$ and for $v$ with $d(v)\le k$, the upper
direction gives $\mathtt{dist}_v\le d(v)$ and the lower direction
gives $\mathtt{dist}_v\ge\min\{d(v),k+1\}=d(v)$.
The parent rule is a function of the distances of the neighbours of
$v$, all of which are at distance at most $d(v)+1$ from $\rho$, so it
is correct one round later.
\end{proof}

\begin{lemma}
\label{lem:B-cycle}
After Layers~1--2 have stabilized, for every $C\in\calC(G)$ with
non-tree edge $e_C=\{a,b\}$ and $c=\mathrm{cid}(C)=\mathrm{key}(a,b)$
we have $\{v \mid c\in\mathtt{Rec}_v\}=V(C)$; moreover
$\mathtt{src}_v[c]=2$ exactly for $v=x_C:=\mathrm{lca}(a,b)$, the
vertex of $C$ of minimum depth, and $\mathtt{Src}_{x_C}[c]$ contains
the two values $\mathtt{dist}_a$ and $\mathtt{dist}_b$.
No identifier is held by a vertex outside the corresponding cycle.
\end{lemma}

\begin{proof}
$V(C)$ is the union of the tree paths $a\rightsquigarrow x_C$ and
$b\rightsquigarrow x_C$ together with $e_C$.
A vertex strictly inside one of these paths has exactly one child
carrying $c$ and is not an endpoint of $e_C$, so
$|\mathtt{Src}|=1$; each of $a,b$ has an anchor and no child carrying
$c$ unless it equals $x_C$; and $x_C$ receives $c$ from its two
distinct children on the two paths or---when $x_C\in\{a,b\}$---from
one child and from its own anchor, so $|\mathtt{Src}|=2$ in either
case.
Forwarding stops exactly at $x_C$, so no vertex outside $V(C)$
obtains $c$, and $\mathtt{bot}$ copies the depth of the endpoint at
the bottom of each chain upward unchanged.
Finally $c\in\mathtt{Cand}_v$ requires an anchor or a forwarding
child, so a fabricated identifier disappears in one round.
\end{proof}

\begin{example}
Let $C$ be the $4$-cycle $v\,a\,b\,c$ with a pendant leaf at $v$ and
$r=v$, so that the stabilized depths are $0,1,2,1$.
The unique non-tree edge is $\{b,c\}$ and $c=\mathrm{key}(b,c)$; the
wave travels $b\to a\to v$ and $c\to v$, and $v$ has
$|\mathtt{Src}_v[c]|=2$, so $x_C=v$ and all four vertices hold the
record, with
$\mathtt{size}=2+1-2\cdot0+1=4$ by the rule for $\mathtt{size}$ in
Algorithm~\ref{alg:B2}.
Under a naming scheme based on the minimum-depth vertex of $C$ the
wave could not even have been started, since $x_C$ is precisely what
the wave computes.
\end{example}

\subsubsection*{Layers 3--4: cycle geometry and articulation points}

\begin{algorithm}[tb]
\caption{Rules for process $v$ --- Layer~3 (geometry of $C$) and Layer~4 (articulation points)}
\label{alg:B2}
\begin{algorithmic}[1]
\Statex \textbf{Layer 3 (evaluation order: $\calB_r$ from the root,
        depth $\le D$)}
\ForAll{$c \in \mathtt{Rec}_v$}
  \If{$\mathtt{src}_v[c]=2$}
     \Comment{$v=x_C$}
     \State let $\mathtt{Src}_v[c]=\{(\sigma_0,\beta_0),
            (\sigma_1,\beta_1)\}$ ordered by
            $\bigl(\beta_i,\mathrm{id}(\sigma_i)\bigr)$, with
            $\mathrm{id}(\bot):=-\infty$
     \State $\mathtt{distx}_v[c] \gets \mathtt{dist}_v$;\quad
            $\mathtt{size}_v[c] \gets
             \beta_0+\beta_1-2\,\mathtt{dist}_v+1$;\quad
            $\mathtt{tau}_v[c] \gets
             \lfloor(\mathtt{size}_v[c]-1)/2\rfloor$
     \State $\mathtt{arm0}_v[c] \gets \sigma_0$;\quad
            $\mathtt{arm}_v[c] \gets \bot$;\quad
            $\mathtt{pos}_v[c] \gets 0$
  \Else
     \State $p \gets \mathtt{par}_v$
     \State $\mathtt{distx}_v[c] \gets \mathtt{distx}_p[c]$;\quad
            $\mathtt{size}_v[c] \gets \mathtt{size}_p[c]$;\quad
            $\mathtt{tau}_v[c] \gets \mathtt{tau}_p[c]$
     \State \textbf{if} $\mathtt{src}_p[c]=2$ \textbf{then}
            $\mathtt{arm}_v[c] \gets
             \bigl(\mathtt{arm0}_p[c]=v \,?\, 0 : 1\bigr)$
            \textbf{else}
            $\mathtt{arm}_v[c] \gets \mathtt{arm}_p[c]$
     \State $\mathtt{pos}_v[c] \gets
            \mathtt{dist}_v-\mathtt{distx}_v[c]$ \textbf{if}
            $\mathtt{arm}_v[c]=0$, \textbf{else}
            $\mathtt{size}_v[c]-
             \bigl(\mathtt{dist}_v-\mathtt{distx}_v[c]\bigr)$
  \EndIf
  \State $\mathrm{prv}_v[c] \gets$ the $u\in N(v)$ with
         $c\in\mathtt{Rec}_u$ and
         $\mathtt{pos}_u[c]\equiv\mathtt{pos}_v[c]-1
          \pmod{\mathtt{size}_v[c]}$
  \State $\mathrm{nxt}_v[c] \gets$ the $u\in N(v)$ with
         $c\in\mathtt{Rec}_u$ and
         $\mathtt{pos}_u[c]\equiv\mathtt{pos}_v[c]+1
          \pmod{\mathtt{size}_v[c]}$
\EndFor
\Statex
\Statex \textbf{Layer 4 (local)}
\State $\mathtt{isAP}_v \gets
       \bigl(\delta(v)\ge3\bigr) \vee
       \bigl(\delta(v)=2 \wedge \mathtt{Rec}_v=\emptyset\bigr)$
\end{algorithmic}
\end{algorithm}

A direct computation from Lemma~\ref{lem:B-cycle} shows that
$\mathtt{pos}$ enumerates $V(C)$ as $w_0=x_C,w_1,\dots,w_{|C|-1}$
along $C$, consecutive positions being adjacent---the positions
$\beta_0-\mathtt{dist}_{x_C}$ and $\beta_0-\mathtt{dist}_{x_C}+1$
being joined by $e_C$ itself---so that $\mathrm{prv}$ and
$\mathrm{nxt}$ are well defined and are the two neighbours of $v$ on
$C$.
Note that $|C|$ is obtained \emph{arithmetically} from three BFS
depths; at no point does a rule range over ``all vertices carrying the
same identifier''.

\begin{lemma}
\label{lem:B-x-ap}
$x_C$ is an articulation point of $G$; consequently no run of $C$
contains $w_0$, and the runs of $C$ are intervals of the path
$w_1,\dots,w_{|C|-1}$.
\end{lemma}

\begin{proof}
If $r\notin V(C)$ then $x_C$ separates $C$ from $r$; if $r\in V(C)$
then $x_C=r$ and $\delta(r)\ge3$, which in a cactus forces $r$ to be
an articulation point.
\end{proof}

\subsubsection*{Layers 5--8: runs and $\lambda_C$}

\begin{algorithm}[tb]
\caption{Rules for process $v$ --- Layers~5--8 (runs, $\lambda_C$)}
\label{alg:B3}
\begin{algorithmic}[1]
\ForAll{$c \in \mathtt{Rec}_v$}
  \Statex \quad \textbf{Layer 5 (order: $C$ by increasing
          $\mathtt{pos}$, depth $\le|C|$)}
  \State $\mathtt{lrun}_v[c] \gets 0$ \textbf{if}
         $\mathtt{isAP}_v \vee \mathtt{pos}_v[c]=0$, \textbf{else}
         $1+\mathtt{lrun}_{\mathrm{prv}_v[c]}[c]$
  \Statex \quad \textbf{Layer 6 (order: $C$ by decreasing
          $\mathtt{pos}$, depth $\le|C|$)}
  \State $\mathtt{rrun}_v[c] \gets 0$ \textbf{if}
         $\mathtt{isAP}_v \vee \mathtt{pos}_v[c]=0$, \textbf{else}
         $1+\mathtt{rrun}_{\mathrm{nxt}_v[c]}[c]$
  \State $\mathtt{run}_v[c] \gets
         \mathtt{lrun}_v[c]+\mathtt{rrun}_v[c]-1$ \textbf{if}
         $\neg\mathtt{isAP}_v \wedge \mathtt{pos}_v[c]\ge1$,
         \textbf{else} $0$
  \Statex \quad \textbf{Layer 7 (order: $C$ by increasing
          $\mathtt{pos}$, depth $\le|C|$)}
  \State $\mathtt{pmax}_v[c] \gets 0$ \textbf{if}
         $\mathtt{pos}_v[c]=0$, \textbf{else}
         $\max\bigl\{\mathtt{run}_v[c],
          \mathtt{pmax}_{\mathrm{prv}_v[c]}[c]\bigr\}$
  \Statex \quad \textbf{Layer 8 (order: $\calB_r$ from the root,
          depth $\le D$)}
  \State $\mathtt{Lam}_v[c] \gets
         \mathtt{pmax}_{\mathrm{prv}_v[c]}[c]$ \textbf{if}
         $\mathtt{pos}_v[c]=0$, \textbf{else}
         $\mathtt{Lam}_{\mathtt{par}_v}[c]$
         \Comment{$\mathrm{prv}$ of $x_C$ is $w_{|C|-1}$}
\EndFor
\end{algorithmic}
\end{algorithm}

By Lemma~\ref{lem:B-x-ap} the recursions of Layers~5--7 are grounded
at $w_0$, so they are genuine path evaluations and
$\mathtt{run}_v[c]$ is the length of the run of $C$ containing $v$.
The global maximum $\lambda_C$ is obtained as a running maximum and is
read by $x_C$ from its neighbour at position $|C|-1$, then broadcast
downwards; no other form of aggregation is used.

\subsubsection*{Layers 9--13: beards, witnesses and deactivation}

\begin{algorithm}[tb]
\caption{Rules for process $v$ --- Layers~9--13 (beards, witnesses, deactivation)}
\label{alg:B4}
\begin{algorithmic}[1]
\Statex \textbf{Layer 9 (order: $\calB_r$ from the leaves, depth
        $\le D$)}
\State $\mathtt{chain}_v \gets \textsf{true}$ \textbf{if}
       $\delta(v)=1$; \textbf{else if}
       $\mathtt{Rec}_v=\emptyset \wedge \delta(v)=2 \wedge
        |\mathrm{ch}(v)|=1$ \textbf{then}
       $\mathtt{chain}_v\gets\mathtt{chain}_u$ for the unique
       $u\in\mathrm{ch}(v)$; \textbf{else}
       $\mathtt{chain}_v\gets\textsf{false}$
\State $\mathrm{TC}_v \gets \{u\in\mathrm{ch}(v) \mid
       \forall c\in\mathtt{Rec}_v\cap\mathtt{Rec}_u:
       u\notin\{\mathrm{prv}_v[c],\mathrm{nxt}_v[c]\}\}$
       \Comment{tree children of $v$}
\Statex
\Statex \textbf{Layers 10--11 (local, then $C$ by increasing
        $\mathtt{pos}$ and $\calB_r$ from the root)}
\ForAll{$c \in \mathtt{Rec}_v$}
  \State $\mathtt{bBnd}_v[c] \gets
         \bigl(|\mathtt{Rec}_v|=1\bigr) \wedge
         \bigl(\delta(v)=3\bigr) \wedge
         \bigl(|\mathrm{TC}_v|=1\bigr) \wedge
         \mathtt{chain}_u \wedge
         \bigl(\mathtt{pos}_v[c]\neq0 \vee v=r\bigr)$,
         $\mathrm{TC}_v=\{u\}$
  \State $\mathtt{alpha}_v[c] \gets
         \mathtt{run}_{\mathrm{prv}_v[c]}[c]$;\quad
         $\mathtt{beta}_v[c] \gets
         \mathtt{run}_{\mathrm{nxt}_v[c]}[c]$
  \State $\mathtt{lamDe}_v[c] \gets
         \max\bigl\{\mathtt{Lam}_v[c],\
         \mathtt{alpha}_v[c]+1+\mathtt{beta}_v[c]\bigr\}$
         \Comment{$=\lambda_C^{v}$}
  \State $\mathtt{wit}_v[c] \gets \mathtt{bBnd}_v[c] \wedge
         \bigl(\mathtt{Lam}_v[c]<\mathtt{tau}_v[c]\bigr) \wedge
         \bigl(\mathtt{lamDe}_v[c]>\mathtt{tau}_v[c]\bigr)$
  \State $\mathtt{wmin}_v[c] \gets \infty$ \textbf{if}
         $\mathtt{pos}_v[c]=0$, \textbf{else}
         $\min\bigl\{\pi_v[c],\
          \mathtt{wmin}_{\mathrm{prv}_v[c]}[c]\bigr\}$ where
         $\pi_v[c]=\mathtt{pos}_v[c]$ if $\mathtt{wit}_v[c]$ and
         $\infty$ otherwise
  \State $\mathtt{wpos}_v[c] \gets
         \mathtt{wmin}_{\mathrm{prv}_v[c]}[c]$ \textbf{if}
         $\mathtt{pos}_v[c]=0$, \textbf{else}
         $\mathtt{wpos}_{\mathtt{par}_v}[c]$
  \State $\mathtt{deact}_v[c] \gets \mathtt{wit}_v[c] \wedge
         \bigl(\mathtt{pos}_v[c]=\mathtt{wpos}_v[c]\bigr)$
         \Comment{exactly one witness per cycle}
\EndFor
\Statex
\Statex \textbf{Layer 12 (order: $\calB_r$ from the root, depth
        $\le D$)}
\State $\mathtt{deAbv}_v \gets \textsf{false}$ \textbf{if}
       $\mathtt{par}_v=\bot \vee \mathtt{Rec}_v\neq\emptyset$;
       \textbf{else if} $\mathtt{Rec}_{\mathtt{par}_v}\neq\emptyset$
       \textbf{then} $\mathtt{deAbv}_v \gets
       \bigl(\exists c\in\mathtt{Rec}_{\mathtt{par}_v}:
       \mathtt{deact}_{\mathtt{par}_v}[c]\bigr)$;
       \textbf{else} $\mathtt{deAbv}_v \gets
       \mathtt{deAbv}_{\mathtt{par}_v}$
\Statex
\Statex \textbf{Layer 13 (runs after deactivation: Layers~5--8 with
        $\mathtt{isAP}$ replaced by
        $\mathtt{isAP}_v\wedge\neg\mathtt{deact}_v[c]$)}
\State compute $\mathtt{lrun}'_v[c],\mathtt{rrun}'_v[c],
       \mathtt{run}'_v[c],\mathtt{pmax}'_v[c],\mathtt{Lam}'_v[c]$
       by the rules of Algorithm~\ref{alg:B3} with that replacement
\end{algorithmic}
\end{algorithm}

The clause $\mathtt{pos}_v[c]\neq0 \vee v=r$ in the rule for
$\mathtt{bBnd}$ excludes
$x_C\neq r$, whose branch contains $r$ and is therefore not a beard;
the clause $\delta(v)=3$ together with $|\mathtt{Rec}_v|=1$ encodes
(B2), and $\mathtt{chain}_u$ encodes (B1) and~(B3).
Selecting the witness of smallest \emph{position} rather than of
smallest identifier makes the tie-break depend only on the leader and
on $\calB_r$.

\subsubsection*{Layers 14--15: classification and output}

\begin{algorithm}[tb]
\caption{Rules for process $v$ --- Layers~14--15 (classification and output)}
\label{alg:B5}
\begin{algorithmic}[1]
\Statex \textbf{Layer 14 (order: $C$ by increasing $\mathtt{pos}$,
        then $\calB_r$ from the root)}
\ForAll{$c \in \mathtt{Rec}_v$}
  \State $\mathtt{fmin}_v[c] \gets \infty$ \textbf{if}
         $\mathtt{pos}_v[c]=0$, \textbf{else}
         $\min\bigl\{\varphi_v[c],\
          \mathtt{fmin}_{\mathrm{prv}_v[c]}[c]\bigr\}$ where
         $\varphi_v[c]=\mathtt{pos}_v[c]$ if
         $\neg\mathtt{isAP}_v\wedge\mathtt{lrun}_v[c]=1\wedge
          \mathtt{run}_v[c]=\mathtt{Lam}_v[c]$, and $\infty$
         otherwise
  \State $\mathtt{fpos}_v[c] \gets
         \mathtt{fmin}_{\mathrm{prv}_v[c]}[c]$ \textbf{if}
         $\mathtt{pos}_v[c]=0$, \textbf{else}
         $\mathtt{fpos}_{\mathtt{par}_v}[c]$
  \If{$\mathtt{Lam}_v[c]>\mathtt{tau}_v[c]$}
       \State $\mathtt{cyctype}_v[c] \gets \textsf{twin}$
  \ElsIf{$\mathtt{Lam}_v[c]=\mathtt{tau}_v[c]$}
       \State $\mathtt{cyctype}_v[c] \gets \textsf{unit}$
  \ElsIf{$\mathtt{wpos}_v[c]\neq\infty$}
       \State $\mathtt{cyctype}_v[c] \gets \textsf{latent-twin}$
  \Else
       \State $\mathtt{cyctype}_v[c] \gets \textsf{zero}$
  \EndIf
  \If{$\mathtt{cyctype}_v[c]\in\{\textsf{twin},
       \textsf{latent-twin}\}$}
     \State $\mathtt{sel}_v[c] \gets
        \neg\mathtt{isAP}_v \wedge
        \mathtt{run}'_v[c]=\mathtt{Lam}'_v[c] \wedge
        \bigl(\mathtt{lrun}'_v[c]=1 \vee \mathtt{rrun}'_v[c]=1\bigr)$
  \ElsIf{$\mathtt{cyctype}_v[c]=\textsf{unit}$}
     \State $\mathtt{sel}_v[c] \gets
        \neg\mathtt{isAP}_v \wedge
        \mathtt{run}_v[c]=\mathtt{Lam}_v[c] \wedge
        \mathtt{lrun}_v[c]=1 \wedge
        \mathtt{pos}_v[c]=\mathtt{fpos}_v[c]$
  \Else
     \State $\mathtt{sel}_v[c] \gets \textsf{false}$
  \EndIf
\EndFor
\Statex
\Statex \textbf{Layer 15 (local): output}
\State $\mathtt{inMVS}_v \gets
       \bigl(\delta(v)=1 \wedge \neg\mathtt{deAbv}_v\bigr) \vee
       \bigl(\exists c\in\mathtt{Rec}_v: \mathtt{sel}_v[c]\bigr)$
\end{algorithmic}
\end{algorithm}

For a twin- or latent twin-cycle the two selected vertices are the two
endpoints of the unique longest deactivated run
(Lemma~\ref{lem:unique-run}), which are exactly the vertices of that
run with $\mathtt{lrun}'=1$ or $\mathtt{rrun}'=1$; the witness
$v^{*}_C$ itself is excluded by $\neg\mathtt{isAP}_v$, and is in any
case internal to the run by Lemma~\ref{lem:unique-run}.
For a unit-cycle, the rule for $\mathtt{fmin}$ marks the first vertex
of a longest run and that for $\mathtt{fpos}$ broadcasts the smallest such position, so exactly one of the
at most two candidate runs is used.

\subsection{Correctness and Complexity of Algorithm~1}
\label{sec:alg1-analysis}

We first show that the rules of
Rule~Sets~\ref{alg:B1}--\ref{alg:B5} have a unique fixed point and
that this fixed point is the canonical set $\Sstar$ of
Definition~\ref{def:canonical}; the round and space bounds announced
in Section~\ref{sec:alg1} then follow.

\begin{theorem}
\label{thm:B-main}
The rules of Rule~Sets~\ref{alg:B1}--\ref{alg:B5} admit exactly one
legitimate configuration, that is, exactly one configuration in which
no process is enabled; its output is
$\{v \mid \mathtt{inMVS}_v\}=\Sstar$, a maximum \MVS\ of $G$
containing no articulation point.
Moreover, under the weakly fair distributed daemon every execution
reaches that configuration within $O(D)$ rounds, and the algorithm is
silent and uses $O\bigl((\gamma_v+1)\log n\bigr)$ bits per process
after stabilization.
\end{theorem}

\begin{proof}
\emph{Uniqueness of the legitimate configuration.}
In a configuration in which no process is enabled, every variable
equals the value prescribed by its rule.
The rule of Layer~1 pins $\mathtt{dist}_r=0$ and is the
distance-plus-one rule elsewhere, whose only fixed point over the
finite domain $\{0,\dots,N\}$ is $\mathtt{dist}_v=d_G(r,v)$
(Lemma~\ref{lem:bfs-time}), together with the tie-broken parent.
Layers~2--15 satisfy the hypotheses of Lemma~\ref{lem:B-layer}: the
rule of a variable of layer~$i$ refers only to variables of layers
$j<i$ at $v$ or at a neighbour of $v$, and to variables of layer~$i$
at in-neighbours of $v$ in an acyclic evaluation order.
By induction on the layer and, inside a layer, on the position in the
evaluation order, the value of every variable in such a configuration
is therefore uniquely determined by $G$.
Hence there is exactly one legitimate configuration.

\emph{The legitimate configuration has the intended values.}
We follow the layers.  Throughout, $C$ denotes a cycle component with
non-tree edge $e_C=\{a,b\}$ and $c=\mathrm{cid}(C)=\mathrm{key}(a,b)$.

\smallskip\noindent
\emph{Layers 1--2 (BFS tree, cycle membership).}
Lemma~\ref{lem:bfs-time} identifies the values of Layer~1, and
Lemma~\ref{lem:B-cycle} then gives
$\{v \mid c\in\mathtt{Rec}_v\}=V(C)$, the identification of $x_C$ as
the unique vertex of $C$ with $\mathtt{src}[c]=2$, and the fact that
no vertex outside $V(C)$ holds $c$.
In particular $\mathtt{Rec}_v$ is in bijection with the set of cycles
through $v$, so $|\mathtt{Rec}_v|=\gamma_v$.

\smallskip\noindent
\emph{Layer 3 (geometry of $C$).}
Since $x_C=\mathrm{lca}(a,b)$ and both $a\rightsquigarrow x_C$ and
$b\rightsquigarrow x_C$ are tree paths, we have
$|C| = (\mathtt{dist}_a-\mathtt{dist}_{x_C}) +
(\mathtt{dist}_b-\mathtt{dist}_{x_C}) + 1$, which is the value
assigned to $\mathtt{size}$; hence $\mathtt{tau}=\tau_C$ as well.
Walking down the two tree paths from $x_C$, the rule for
$\mathtt{pos}$ enumerates $V(C)$ as $w_0=x_C,w_1,\dots,w_{|C|-1}$
along $C$, the arm labelled $0$ receiving the positions
$1,2,\dots$ and the arm labelled $1$ the positions
$|C|-1,|C|-2,\dots$; consecutive positions are adjacent, the two
positions on either side of $e_C$ being joined by $e_C$ itself.
Consequently $\mathrm{prv}_v[c]$ and $\mathrm{nxt}_v[c]$ are well
defined and are exactly the two neighbours of $v$ on $C$.
Note that $|C|$ is obtained \emph{arithmetically} from three BFS
depths; at no point does a rule range over ``all vertices carrying the
same identifier''.

\smallskip\noindent
\emph{Layer 4 ($\mathtt{isAP}$ computes $\calA(G)$).}
Let $\delta(v)\ge3$.  The cycles through $v$ are edge-disjoint and
each of them uses exactly two of the edges incident to $v$, so either
some edge at $v$ lies on no cycle --- and is then a bridge, whose
removal together with $v$ disconnects its other endpoint --- or $v$
lies on at least two cycles, which $v$ separates.
In both cases $v\in\calA(G)$.
Let $\delta(v)=2$.  If $v$ lies on a cycle $C$, then $C-v$ is a path
joining the two neighbours of $v$, and every other vertex reaches $C$
without passing through $v$, so $v\notin\calA(G)$; if $v$ lies on no
cycle, both its edges are bridges and $v\in\calA(G)$.
By Lemma~\ref{lem:B-cycle} the condition ``$v$ lies on no cycle'' is
exactly $\mathtt{Rec}_v=\emptyset$.
Finally a leaf is never an articulation point, and the rule assigns
$\textsf{false}$ to it.
These three cases are precisely the rule of Layer~4.

\smallskip\noindent
\emph{Layers 5--8 ($\mathtt{run}$ and $\mathtt{Lam}=\lambda_C$).}
By Lemma~\ref{lem:B-x-ap} the vertex $w_0=x_C$ is an articulation
point, so no $\emptyset$-free run of $C$ contains $w_0$ and every run
is an interval of the path $w_1,\dots,w_{|C|-1}$.
The recursions of Layers~5 and~6 are therefore grounded at $w_0$ and
are genuine path evaluations: an immediate induction gives that
$\mathtt{lrun}_v[c]$ (resp.\ $\mathtt{rrun}_v[c]$) is the number of
vertices of the run containing $v$ that lie between the beginning
(resp.\ the end) of that run and $v$, inclusive.
Hence $\mathtt{run}_v[c]=\mathtt{lrun}_v[c]+\mathtt{rrun}_v[c]-1$ is
the length of the run containing $v$, and it is $0$ when $v$ is a
boundary vertex or $v=w_0$.
Layer~7 accumulates the maximum of $\mathtt{run}$ along
$w_1,\dots,w_{|C|-1}$, so $\mathtt{pmax}_{w_{|C|-1}}[c]=\lambda_C$;
Layer~8 lets $x_C$ read that value from $\mathrm{prv}_{x_C}[c]=
w_{|C|-1}$ and broadcasts it along $\calB_r$, so
$\mathtt{Lam}_v[c]=\lambda_C$ for every $v\in V(C)$.

\smallskip\noindent
\emph{Layers 9--11 (beards, $\lambda_C^{v}$, witnesses).}
The rule for $\mathtt{chain}$ is a leaf-to-root evaluation along
$\calB_r$ whose fixed point is: $\mathtt{chain}_v=\textsf{true}$ if
and only if the subtree of $\calB_r$ rooted at $v$ is a path whose
vertices all have degree two in $G$ and lie on no cycle, except its
bottom vertex, which is a leaf of $G$.
This is exactly (B1) and~(B3) for the part of the graph below $v$.
For a boundary vertex $v$ of $C$, the branch $G_v$ is a beard if and
only if, in addition, $v$ lies on exactly one cycle and carries
exactly one edge outside that cycle, and $G_v$ does not contain the
root: these are the clauses $|\mathtt{Rec}_v|=1$, $\delta(v)=3$ and
$|\mathrm{TC}_v|=1$ with $\mathtt{chain}_u$, and the clause
$\mathtt{pos}_v[c]\neq0 \vee v=r$, which excludes $v=x_C\neq r$, whose
branch contains $r$ and is therefore not a beard.
Hence $\mathtt{bBnd}_v[c]$ is the predicate ``$v$ is a beard boundary
vertex of $C$'' of Definition~\ref{def:cycle-types}.
Since $\mathtt{run}$ vanishes on boundary vertices, $\mathtt{alpha}$
and $\mathtt{beta}$ are the lengths $\alpha_v,\beta_v$ of the two
$\emptyset$-free runs adjacent to $v$, so
$\mathtt{lamDe}_v[c]=\max\{\lambda_C,\alpha_v+1+\beta_v\}=
\lambda_C^{v}$ by the computation in the proof of
Lemma~\ref{lem:unique-run}, and $\mathtt{wit}_v[c]$ is the predicate
``$v$ is a witness of the latent twin-cycle $C$''.
Finally $\mathtt{wmin}$ is a running minimum of the positions of the
witnesses, $\mathtt{wpos}$ broadcasts the smallest of them, and
$\mathtt{deact}$ marks exactly one witness per cycle.

\smallskip\noindent
\emph{Layers 12--13 (deactivation).}
The rule for $\mathtt{deAbv}$ is a top-down evaluation along
$\calB_r$ that propagates the deactivation flag of a boundary vertex
to the vertices of the tree components hanging below it and stops at
the next cycle; since the beard $G_{v^{*}_C}$ lies below $v^{*}_C$ in
$\calB_r$, its leaf $z_C$ is exactly the vertex with
$\delta=1$ and $\mathtt{deAbv}=\textsf{true}$.
Layer~13 repeats Layers~5--8 with $\mathtt{isAP}$ replaced by
$\mathtt{isAP}\wedge\neg\mathtt{deact}$, \ie with the selected witness
deactivated, so $\mathtt{run}'$ and $\mathtt{Lam}'$ are the run
lengths and the maximum run length of the $X_C$-free runs, where
$X_C=\emptyset$ for a twin-cycle and $X_C=\{v^{*}_C\}$ for a latent
twin-cycle.

\smallskip\noindent
\emph{Layers 14--15 (classification and output).}
The four cases of the rule for $\mathtt{cyctype}$ are, in order,
$\lambda_C>\tau_C$, $\lambda_C=\tau_C$, $\lambda_C<\tau_C$ with a
witness, and $\lambda_C<\tau_C$ without one: this is verbatim
Definition~\ref{def:cycle-types}.
For a twin- or latent twin-cycle the longest $X_C$-free run is unique
(Lemma~\ref{lem:unique-run}) and has length
$\lambda_C^{X_C}>\tau_C\ge1$, hence at least two vertices; its two
endpoints are exactly the vertices $v$ of that run with
$\mathtt{lrun}'_v[c]=1$ or $\mathtt{rrun}'_v[c]=1$, and they are not
articulation points, again by Lemma~\ref{lem:unique-run}.
The conjunct $\neg\mathtt{isAP}_v$ therefore excludes only the witness
$v^{*}_C$, which is internal to the run.
For a unit-cycle, $\mathtt{fmin}$ marks the first vertex of each
longest run and $\mathtt{fpos}$ broadcasts the smallest such position,
so exactly one of the at most two candidate runs is used and exactly
one vertex of it is selected.
For a zero-cycle nothing is selected.
Layer~15 finally sets $\mathtt{inMVS}_v$ for every leaf whose beard
was not deactivated and for every selected vertex of a cycle.
Comparing with Definition~\ref{def:canonical}, we conclude
$\{v\mid\mathtt{inMVS}_v\}=\Sstar$, with the two tie-breaks resolved
by position instead of by index; since Lemmas~\ref{lem:size} and
\ref{lem:local} and Proposition~\ref{prop:construction} hold for
\emph{every} admissible choice of $P_C$, $u_C$ and $v^{*}_C$, this set
is a maximum \MVS\ by Theorem~\ref{thm:mvn}, and it contains no
articulation point by Lemma~\ref{lem:size}.

\emph{Convergence and silence.}
By Lemma~\ref{lem:bfs-time} the variables of Layer~1 hold their
prescribed values from the end of round $D+1$ on, and the parent
pointers one round later.
Every one of the remaining $14$ layers has depth $O(D)$: the layers
evaluated along $\calB_r$ have depth at most $D$, and the layers
evaluated along a cycle $C$ have depth at most $|C|-1\le2D$, because a
cycle component is an isometric subgraph of $G$
(Observation~\ref{obs:isometric}) and hence
$\lfloor|C|/2\rfloor\le D$.
By Lemma~\ref{lem:B-layer} the system reaches the legitimate
configuration within $O(D)$ further rounds.
Silence is immediate, a legitimate configuration being a fixed point
of the rules.
The space bounds are those established in the list of variables above
and are restated in Proposition~\ref{prop:alg1-space}.
\end{proof}

\begin{proposition}
\label{prop:alg1-time}
Algorithm~1 stabilizes in $O(D)$ rounds under the distributed daemon.
\end{proposition}

\begin{proof}
By Lemma~\ref{lem:bfs-time}, $\mathtt{dist}$ is correct after $D+1$
rounds and $\mathtt{par}$ one round later, so Layer~1 is stable after
$D+2$ rounds.
From that moment Lemma~\ref{lem:B-layer} applies to Layers~2--15.
Their evaluation orders are either $\calB_r$, of depth at most $D$, or
a cycle $C$ traversed by increasing or by decreasing position, of
depth at most $|C|-1$; and $|C|\le2D+1$ because a cycle component is
an isometric subgraph of $G$ (Observation~\ref{obs:isometric}).
Each of the $14$ layers therefore has depth $O(D)$, and their number
is a constant, so Lemma~\ref{lem:B-layer} yields a configuration in
which every variable holds its prescribed value $O(D)$ rounds later.
By Theorem~\ref{thm:B-main} that configuration is the unique
legitimate one, and no process is enabled in it.
The total is $O(D)$ rounds.
\end{proof}

\begin{proposition}
\label{prop:alg1-space}
Algorithm~1 requires $O\bigl((\gamma_v+1) \log n\bigr)$ bits of memory
per process $v$, where $\gamma_v$ is the number of cycles in
$\mathcal{C}(G)$ containing $v$.
The total memory across all processes is $O(n \log n)$ bits, giving
an average of $O(\log n)$ bits per process.
\end{proposition}

\begin{proof}
We use the list of variables given in Section~\ref{sec:rules}.
The part of the state that does not depend on the cycles ---
$\mathtt{par}_v$, $\mathtt{dist}_v$ and the four booleans
$\mathtt{isAP}_v$, $\mathtt{chain}_v$, $\mathtt{deAbv}_v$,
$\mathtt{inMVS}_v$ --- occupies $O(\log n)$ bits, because
$\mathtt{par}_v$ is the identifier of a neighbour and
$\mathtt{dist}_v$ ranges over $\{0,\dots,N\}$ with $N=n^{O(1)}$.
Every field of a record is an identifier pair, a value of
$\{0,\dots,N\}\cup\{\infty\}$, or a constant-size flag, and the number
of fields is a constant, so a record occupies $O(\log n)$ bits.
In an \emph{arbitrary} configuration, $v$ holds one record per element
of $\mathtt{Cand}_v$, and the rule for $\mathtt{Cand}_v$ in
Rule~Set~\ref{alg:B1} gives
$|\mathtt{Cand}_v|\le|E^{\mathrm{nt}}_v|+|\mathrm{ch}(v)|\le\delta(v)$;
after stabilization $\mathtt{Rec}_v$ consists exactly of the
$\gamma_v$ cycles through $v$ by Lemma~\ref{lem:B-cycle}.
Hence $v$ uses $O\bigl((\delta(v)+1)\log n\bigr)$ bits at all times
and $O\bigl((\gamma_v+1)\log n\bigr)$ bits after stabilization.
A cactus graph with $k$ cycles has exactly $n-1+k$ edges and
$k\le(n-1)/2$, so
$\sum_{v\in V}\delta(v)=2|E|\le3(n-1)$ and
\[
   \sum_{v\in V}\gamma_v=\sum_{C\in\calC(G)}|C|
   \;\le\; |E| \;=\; n+k-1 \;=\; O(n).
\]
Both bounds therefore give
$\sum_{v\in V}O\bigl((\gamma_v+1)\log n\bigr)
= O\bigl(\log n\cdot(n+\sum_v\gamma_v)\bigr)=O(n\log n)$ bits in
total, that is, an average of $O(\log n)$ bits per process.
\end{proof}

\begin{remark}
\label{rem:space-lb}
Algorithm~1 is silent: once it has stabilized, no communication
register changes.
We deliberately do \emph{not} claim that its average space is optimal.
The $\Omega(\log n)$ bound of Dolev, Gouda and
Schneider~\cite{DolGouSch1999} applies to silent algorithms whose
registers allow a spanning structure of the network to be
reconstructed, and a maximum \MVS\ does not encode one, so the bound
does not transfer to the problem studied here; it does apply to any
algorithm following our strategy, whose registers encode
$\calB_r$, but that is a statement about the strategy and not about
the problem.
No non-trivial space lower bound is known for the maximum \MVS\
problem, and establishing one is left open
(Section~\ref{sec:conclusion}).
The worst case of Algorithm~1 also exceeds its average: a process on
$\gamma_v$ cycles stores $\Theta(\gamma_v\log n)$ bits, and $\gamma_v$
may be $\Theta(n)$ (\eg at the centre of a friendship graph).
\end{remark}

\begin{remark}
Three design points deserve emphasis.
First, identifying a cycle with its unique non-tree edge removes the
circularity of naming a cycle after its minimum-depth vertex: the
identifier is computable by the two endpoints of that edge alone, and
the minimum-depth vertex is a \emph{conclusion} of Layer~2 rather than
a prerequisite for it.
Second, $|C|$ is obtained arithmetically from three BFS depths, so no
rule ever ranges over all vertices carrying a given identifier; the
remaining aggregations are running maxima and minima along the path
$w_1,\dots,w_{|C|-1}$, which are ordinary layers.
Third, all variables are recomputed by total assignments over finite
domains, so a corrupted value is overwritten within one activation and
a fabricated cycle identifier disappears as soon as its anchor is
re-evaluated.
\end{remark}

\subsection{Algorithm 2: A Component-Local Self-Stabilizing Algorithm}
\label{sec:alg2}
Algorithm~2 replaces the single-root BFS of Algorithm~1 by a
\emph{parallel multi-root BFS}, at the cost of substantially increased
memory.
The key insight is that the decisions for a cycle component and its
adjacent tree components can be taken as soon as the relevant local
information has propagated \emph{within} those components, without
waiting for a global BFS tree to stabilize.

Every process $v$ with $\deg(v)\ge3$ acts as a \emph{local root} and
constructs a BFS tree $\mathcal{B}_v$ rooted at itself, with tree
identifier $\mathrm{bfsid}(v)=\mathrm{id}(v)$; every process
participates in all of them concurrently, keeping one BFS-distance
entry per root in
$V^{\delta 3+}=\{v \in V \mid \deg(v)\ge3\}$.
Phases~2--4 are then exactly those of Algorithm~1, with every
propagation carried out inside the local trees $\mathcal{B}_v$ instead
of a single global tree.
Since every cycle component contains a vertex of $V^{\delta 3+}$ ---
a boundary vertex of a cycle has degree at least three --- cycle
detection and the computation of $\lambda_C,\lambda_C^{v}$ inside a
cycle component $C$ complete in $O(|C|)$ rounds.
Crucially, Phase~4 is triggered independently for each component $C$ as
soon as Phases~2 and~3 have completed for $C$ and its adjacent tree
components, without waiting for the rest of the graph.

\subsubsection*{Instances, local roots and the rules of Algorithm~2}

Concretely, Algorithm~2 runs one \emph{instance} of the machinery of
Section~\ref{sec:rules} for every $\rho\in V^{\delta3+}$: in the
instance of $\rho$ the rules are those of
Rule~Sets~\ref{alg:B1}--\ref{alg:B5} with the leader $r$ replaced by
$\rho$.
Each instance is by itself an execution of Algorithm~1 with a
different leader, so Theorem~\ref{thm:B-main} applies to it verbatim;
note that $\delta(\rho)\ge3$ makes $\rho$ an articulation point (see
the proof of Theorem~\ref{thm:B-main}), so that
Lemma~\ref{lem:B-x-ap} holds in every instance.

The point of the replication is that a cycle component can be
processed inside an instance whose root lies \emph{on} it, and is then
not delayed by the rest of the graph.
Every cycle component $C$ contains a vertex of $V^{\delta3+}$: as $G$
is connected and is not a simple cycle, $C$ has a boundary vertex, and
a boundary vertex of a cycle has degree at least three.
We let $\rho_C$ be the vertex of smallest identifier in
$V^{\delta3+}\cap V(C)$ and read the answer for $C$ off the instance
of $\rho_C$.
Two points make this well defined and local.
First, in the instance of a root $\rho$ the minimum-depth vertex of
$C$ is $\rho$ itself if and only if $\rho\in V(C)$, that is, if and
only if $\mathtt{distx}_v[c]=0$ for the record of $C$ in that
instance; since $\mathtt{distx}$ is broadcast to all of $V(C)$ in
Layer~3, every vertex of $C$ sees the same set
$V^{\delta3+}\cap V(C)$ of candidate roots and hence selects the same
$\rho_C$.
Second, a process must recognise, across instances, which records
refer to the same cycle, and the identifier $\mathrm{cid}$ does depend
on the instance, since which edge of $C$ is a non-tree edge depends on
the root.
The matching is nevertheless local: in every instance the record of
$C$ at $v$ determines the unordered pair
$\{\mathrm{prv}_v[c],\mathrm{nxt}_v[c]\}$ of the two neighbours of $v$
on $C$ (Layer~3), that pair does not depend on the instance, and two
distinct cycles through $v$ yield two disjoint pairs.
Layers~12--15 are then evaluated in the selected instance only: the
vertices of $C$, and the vertices of the tree components attached to
$C$, use the values computed in the instance of $\rho_C$.
Restricting the deactivation layers to one instance is what keeps the
tie-break consistent --- a latent twin-cycle may have several
witnesses, and two instances may select different ones, so taking the
disjunction of the deactivation flags over all instances could forfeit
two leaves instead of one.
The rule for $\mathtt{deAbv}$ is otherwise unchanged: it is run inside
$\calB_{\rho_C}$, which is legitimate because $\rho_C\in V(C)$, so
every tree component attached to $C$ lies below its boundary vertex in
$\calB_{\rho_C}$.


\subsubsection*{Complexity analysis of Algorithm~2}

\begin{proposition}
\label{prop:alg2-time}
Algorithm~2 stabilizes in
\[
   \BigO\Bigl(\max_{C\in\calC(G)}|C| \;+\;
              \max_{T \text{ tree comp.}}|T|\Bigr)
   \;=\; \BigO\bigl(|C_{\max}| + |T_{\max}|\bigr)
\]
rounds under the distributed daemon.
\end{proposition}

\begin{proof}
Write $k=|C_{\max}|+|T_{\max}|$, fix a cycle component $C$, and put
$\rho=\rho_C$ and $U=V(C)\cup\bigcup\{V(T) : T$ a tree component
attached to $C\}$.
Since $C$ is isometric in $G$ (Observation~\ref{obs:isometric}) and
$\rho\in V(C)$, every $u\in V(C)$ satisfies
$d_G(\rho,u)\le\lfloor|C|/2\rfloor$, and every $u$ in a tree component
$T$ attached to $C$ satisfies
$d_G(\rho,u)\le\lfloor|C|/2\rfloor+|T|\le k$.
By Lemma~\ref{lem:bfs-time} applied to the instance of $\rho$, the
variables $\mathtt{dist}$ of that instance are correct on the ball of
radius $k$ around $\rho$ from the end of round $k+1$ on, and the
parent pointers one round later --- and this holds however large the
diameter of $G$ is.

From that moment Lemma~\ref{lem:B-layer} and
Corollary~\ref{cor:B-local} apply to the variables attached to the
vertices of $U$ in the instance of $\rho$.
Each of the $14$ layers propagates either along $C$, with depth at
most $|C|-1$, or inside a tree component attached to $C$, with depth
at most $|T|$, or along the part of $\calB_{\rho}$ that joins the two,
with depth at most $k$; and none of these evaluation orders leaves the
ball of radius $k$ around $\rho$, because $C$ separates that ball from
the rest of the graph at its boundary vertices.
Every such variable therefore has dependency depth $O(k)$, and
Corollary~\ref{cor:B-local} makes it stable after $O(k)$ rounds.
By Theorem~\ref{thm:B-main} applied to the instance of $\rho$, the
values reached are those prescribed by
Definition~\ref{def:canonical}; Layers~12--15 for $U$ are evaluated in
this instance and add $O(k)$ further rounds.

The argument applies to every cycle component simultaneously, and a
leaf of a tree component attached to no cycle is selected
unconditionally after $O(1)$ rounds, so every process holds its final
output after $O(k)=O(|C_{\max}|+|T_{\max}|)$ rounds.
\end{proof}

\noindent
The bound of Proposition~\ref{prop:alg2-time} is genuinely
\emph{component-local}: it is not a restatement of the $\BigO(D)$
bound of Algorithm~1, and it cannot be improved as a function of
$|C_{\max}|$ and $|T_{\max}|$.
Both points follow from the family described next, whose diameter is
arbitrarily larger than $|C_{\max}|+|T_{\max}|$.

\begin{lemma}
\label{lem:lb-family}
Let $s\ge2$, $t\ge4$, $m\ge3$, and let $\Gamma_{s,t,m}$ be the cactus
graph built as follows: a \emph{spine} of $m$ triangles joined
consecutively by bridges; a cycle $C$ of length $c=4s+3$ with exactly
three boundary vertices $p,v,q$ in this cyclic order, separated by
runs of $\alpha=s+1$, $\beta=s+1$ and $\gamma=2s-2$
non-articulation points; one endpoint of the spine attached to $p$ by
a bridge, a beard $B$ of $t$ vertices attached to $v$, and a star
$K_{1,3}$ attached to $q$.
Write $z$ for the leaf of $B$.
Then $|C_{\max}|=4s+3$, $|T_{\max}|=t$ and $D=\Theta(m+s+t)$, so
$|C_{\max}|+|T_{\max}|$ stays bounded while $D$ grows without bound.
Moreover $C$ is a latent twin-cycle whose unique witness is $v$, and
every maximum \MVS\ of $\Gamma_{s,t,m}$ containing no articulation
point excludes $z$; if the star at $q$ is deleted, then $C$ becomes a
twin-cycle and every such maximum \MVS\ contains $z$.
\end{lemma}

\begin{proof}
The parameters are immediate: $\tau_C=\lfloor(c-1)/2\rfloor=2s+1$ and
$\lambda_C=\max\{\alpha,\beta,\gamma\}=\max\{s+1,2s-2\}<2s+1=\tau_C$,
while $\lambda_C^{v}=\alpha+1+\beta=2s+3>\tau_C$; the branches at $p$
and at $q$ are not beards (the first contains the spine, the second
branches), so $v$ is the only beard boundary vertex and hence the
unique witness, and $C$ is a latent twin-cycle.
Deleting the star makes $q$ a non-articulation point of degree two,
which merges the runs $\beta$ and $\gamma$ into one of length
$\beta+1+\gamma=3s>2s+1=\tau_C$, so $C$ becomes a twin-cycle. 

For the membership claims, recall from the proof of
Theorem~\ref{thm:mvn} that $|S|$ is bounded by the sum of the local
inequalities of Lemma~\ref{lem:exchange}, so a slack of $\delta$ in
the term of a single cycle component yields $|S|\le\capa(G)-\delta$.
Suppose $z\in S$; then $v\notin I_C$ and $I_C\subseteq\{p,q\}$.
If $I_C=\emptyset$, then $\lambda_C^{I_C}=\lambda_C<\tau_C$, so
$|S\cap V(C)|=0$ by Lemma~\ref{lem:cycle-bound}(i) while
$\capa(C)=1$: slack~$1$.
If $q\in I_C$, then $\ccaps(C,q)=3$ (the three leaves of the star)
while $|S\cap V(C)|\le3$ and $\capa(C)=1$: slack~$\ge1$.
If $p\in I_C$, then $\ccaps(C,p)\ge m$, because each triangle of the
spine has two boundary vertices and hence $\lambda=\tau=1$, \ie is a
unit-cycle of capacity~$1$: slack~$\ge m-2\ge0$, and in fact
slack~$\ge1$ once $m\ge3$.
In every case $|S|\le\capa(G)-1<\mu(G)$, so $z\notin S$.
Symmetrically, after deleting the star, suppose $z\notin S$; then
$v\in I_C$ and $\ccaps(C,v)=1$, while
$|S\cap V(C)|\le2=\capa(C)$ by Lemma~\ref{lem:cycle-bound}(iii) when
$p\notin I_C$, and $|S\cap V(C)|\le3$ against
$\capa(C)+\ccaps(C,v)+\ccaps(C,p)\ge2+1+m$ when $p\in I_C$: slack
$\ge1$ in both cases, so $z\in S$.
\end{proof}

\begin{proposition}
\label{prop:alg2-lb}
Consider the task of computing a maximum \MVS\ containing no
articulation point---the task solved by both our algorithms, and
without loss of generality by Lemma~\ref{lem:no-ap}.
Every correct self-stabilizing algorithm for this task needs
$\Omega(s+t)=\Omega(|C_{\max}|+|T_{\max}|)$ rounds on
$\Gamma_{s,t,m}$.
Since $D=\Theta(m+s+t)$ may be arbitrarily larger than
$|C_{\max}|+|T_{\max}|=\Theta(s+t)$, the bound of
Proposition~\ref{prop:alg2-time} is asymptotically tight as a function
of $|C_{\max}|$ and $|T_{\max}|$, and is not implied by any bound in
terms of the diameter.
\end{proposition}

\begin{proof}
Let $G'$ be $\Gamma_{s,t,m}$ with the star at $q$ deleted.
By Lemma~\ref{lem:lb-family}, $z$ outputs $\textsf{false}$ in every
legitimate configuration of $\Gamma_{s,t,m}$ and $\textsf{true}$ in
every legitimate configuration of $G'$, while the two graphs are
isomorphic within distance $r=t+s+1$ of $z$.
Start the two executions from configurations agreeing on that common
ball.
Under the synchronous daemon---one of the schedules the distributed
daemon may produce---the state of $z$ after $\rho\le r$ rounds depends
only on the initial states within distance $\rho$, so $z$ produces the
same output in both, which is impossible once both have stabilized.
Hence more than $r=\Omega(s+t)$ rounds are needed in one of them.
\end{proof}

\begin{remark}
\label{rem:alg2-vs-alg1}
Taking $s,t$ constant---or simply a chain of $m$ triangles with no
attachments---gives cactus graphs with
$|C_{\max}|+|T_{\max}|=\BigO(1)$ and $D=\Theta(n)$, on which
Algorithm~2 stabilizes in $\BigO(1)$ rounds while Algorithm~1 needs
$\Theta(D)$; this is the regime in which Algorithm~2 is worth its
additional memory.
Conversely, $\Omega(|C_{\max}|+|T_{\max}|)$ is \emph{not} a lower
bound on every cactus graph---if $T_{\max}$ is a star attached to a
cycle, its branching is visible at distance one---which is why
Proposition~\ref{prop:alg2-lb} is stated for a family of instances.
The second parameter must be $|T_{\max}|$ and not the largest beard,
since the leaf of a non-beard component must also learn that its
component is not a beard.
\end{remark}

\subsection{Comparison of the Two Algorithms}
\label{sec:comparison}

\begin{proposition}
\label{prop:alg2-space}
Algorithm~2 requires
$O\!\left(|V^{\delta 3+}| \log D + (\gamma_v+1) \log n\right)$
bits of memory per process $v$, where $D$ is the diameter of $G$,
$|V^{\delta 3+}|$ is the number of vertices of degree at least three,
and $\gamma_v$ is the number of cycles containing $v$.
\end{proposition}

\begin{proof}
The accounting is that of Proposition~\ref{prop:alg1-space}, with one
BFS entry per root instead of one.
Each process $v$ maintains one distance entry
$\mathtt{dist}^{\rho}_v$ per root $\rho\in V^{\delta 3+}$; each of
them ranges over $\{0,\dots,N\}$ but is only ever compared with
distances of neighbours of $v$ in the same instance, so $O(\log D)$
bits suffice once out-of-range values are clamped, contributing
$O(|V^{\delta 3+}|\log D)$ bits in total, and the parent pointer of an
instance is recovered from the distances of the neighbours.
For each cycle $C$ through $v$, the process keeps one record --- the
one of the instance of $\rho_C$ --- holding $\lambda_C$, $\tau_C$,
$\mathrm{cid}(C)$, the run and witness fields and the beard flag, each
of $O(\log n)$ bits, hence $O(\gamma_v\log n)$ bits in total.
The output variable $\mathtt{inMVS}_v$ occupies $O(1)$ bits.
Altogether a process uses
$O\bigl(|V^{\delta 3+}|\log D+(\gamma_v+1)\log n\bigr)$ bits.
\end{proof}

\begin{remark}
The memory usage of Algorithm~2 is dominated by the
$O(|V^{\delta 3+}| \log D)$ term, which can be significantly larger
than the $O(\gamma_v \log n)$ per-process cost of Algorithm~1
when the number of high-degree vertices is large.
In contrast, on cactus graphs where $|V^{\delta 3+}|$ is small and
where the largest cycle component and the largest tree component are
much smaller than the diameter---\ie $|C_{\max}|+|T_{\max}| = o(D)$
---Algorithm~2 achieves both a strictly better stabilization time and
a manageable memory footprint, making it preferable to Algorithm~1 in
latency-sensitive applications.
Table~\ref{tab:comparison} summarizes the complexity of both
algorithms.
\end{remark}

\begin{table}[tb]
\centering
\caption{Complexity comparison of the two algorithms.}
\label{tab:comparison}
\begin{tabular}{lcc}
\hline
\textbf{Algorithm} & \textbf{Stabilization time} &
\textbf{Space per process} \\
\hline
 Algorithm~1 (single-root BFS)
   & $O(D)$
   & $O((\gamma_v+1) \log n)$ (average: $O(\log n)$)\\
 Algorithm~2 (component-local)
   & $O(|C_{\max}| + |T_{\max}|)$
   & $\BigO(|V^{\delta 3+}|\log D + (\gamma_v+1)\log n)$ \\
\hline
\end{tabular}
\end{table}

\section{Conclusion and Open Problems}
\label{sec:conclusion}

We determined the mutual visibility number of cactus graphs,
$\mu(G) = 2n_2+n_{\mathrm{lt}}+n_1+\ell(G)$, via a cycle--tree
decomposition and the classification of each cycle component as a
zero-, unit-, twin- or latent twin-cycle, and we proposed two
self-stabilizing algorithms that construct a maximum \MVS\ under the
distributed daemon: Algorithm~1 stabilizes in $O(D)$ rounds with
$O(\log n)$ average space per process, and Algorithm~2 in
$O(|C_{\max}|+|T_{\max}|)$ rounds---a bound that is asymptotically
tight as a function of these two parameters even when they are
$o(D)$ (Proposition~\ref{prop:alg2-lb})---at the cost of increased
memory.

Three open problems suggest themselves.
First, the \MVN\ remains unresolved for classes that generalize cactus
graphs---outerplanar graphs, series-parallel graphs, and graphs of
bounded treewidth---all of which admit structured decompositions that
may be amenable to the analysis developed here.
Second, no non-trivial space lower bound is known for the maximum
\MVS\ problem: as explained in Remark~\ref{rem:space-lb}, the
$\Omega(\log n)$ bound for silent
stabilization~\cite{DolGouSch1999} does not apply, because a maximum
\MVS\ does not encode a spanning structure of the network.
Proving such a bound, and determining whether the worst-case
per-process space $O((\gamma_v+1)\log n)$ of Algorithm~1 can be
reduced to $O(\log n)$, both remain open.



\bibliography{mvs-bib}

\end{document}